\documentclass[a4paper,UKenglish]{article}
\usepackage{graphicx}
\usepackage{amsthm}
\usepackage{thm-restate}
\usepackage[most]{tcolorbox}
\usepackage{todonotes}
\usepackage{lmodern}
\usepackage{mathtools}
\usepackage{fullpage}
\usepackage{amssymb}
\usepackage{amsmath}
\usepackage{amsfonts}
\usepackage{dsfont}
\usepackage{url}
\usepackage{hyperref}
\usepackage[noend]{algorithm2e}
\usepackage[utf8]{inputenc}
\usepackage[T1]{fontenc}
\usepackage{tikz}
\usepackage{mathrsfs} 
\usepackage{centernot}
\usepackage{mdframed}
\usetikzlibrary{automata,positioning}
\usepackage{totcount}
\usepackage{float}
\usepackage{nccmath}
\usepackage{thmtools}
\usepackage[capitalise]{cleveref}

\theoremstyle{plain}
\newtheorem{thm}{Theorem}
\crefname{thm}{theorem}{theorems}

\crefname{prop}{proposition}{propositions}
\newtheorem{ppt}[thm]{Property}
\crefname{ppt}{property}{properties}
\Crefname{ppt}{Property}{Properties}

\crefname{cor}{corollary}{corollaries}
\newtheorem{lem}[thm]{Lemma}
\crefname{lem}{lemma}{lemmas}
\Crefname{lem}{Lemma}{Lemmas}
\newtheorem{lemma}[thm]{Lemma}
\crefname{lemma}{lemma}{lemmas}
\Crefname{lemma}{Lemma}{Lemmas}

\crefname{pblm}{problem}{problems}
\theoremstyle{definition}
\newtheorem{defn}[thm]{\textsf{Definition}}
\crefname{defn}{definition}{definitions}
\theoremstyle{remark}
\newtheorem{req}[thm]{\textsf{Remark}}
\crefname{req}{remark}{remarks}

\newtotcounter{todoctr}

\newcommand{\RR}{\mathbb{R}}
\newcommand{\NN}{\mathbb{N}}

\newcommand{\eps}{\varepsilon}
\newcommand{\es}{\varnothing}
\newcommand{\subs}{\subseteq}

\renewcommand{\bar}[1]{\overline{#1}}

\DeclareMathOperator*{\argmin}{\arg\!\min}

\newcommand{\opt}{\text{opt}}
\newcommand{\OPT}{\text{OPT}}
\newcommand{\Pp}{\mathcal{P}}
\newcommand{\Aa}{\mathcal{A}}
\newcommand{\Ss}{\mathcal{S}}

\newcommand{\Tt}{\mathcal{T}}
\newcommand{\Ll}{\mathcal{L}}

\newcommand{\MR}{\text{MR}}

\newcommand{\ALG}{\mathrm{\textsc{alg}}}
\newcommand{\SOL}{\mathrm{\textsc{sol}}}
\newcommand{\sol}{\mathrm{\textsc{sol}}}
\newcommand{\alg}{\mathrm{\textsc{alg}}}
\newcommand{\apij}{{\alg^{(i+1)}_{\rho, j}}}
\newcommand{\aptj}{{\alg^{(i+2)}_{\rho, j}}}

\crefname{algocfline}{alg.}{algs.}
\Crefname{algocfline}{Algorithm}{Algorithms}
\Crefname{algocf}{Line}{Lines}

\LinesNumbered

\title{A Reusable Framework for Robust Approximation Algorithms in the Interval Uncertainty Model}

\author{
\hspace{-2cm}
\begin{minipage}{8cm}
    \begin{center}
        \textbf{Ralf Klasing}\\
        CNRS, LaBRI, University of Bordeaux\\
        \texttt{ralf.klasing@labri.fr}
    \end{center}
\end{minipage}\begin{minipage}{6cm}
    \begin{center}
        \textbf{Tobias Mömke}\\
        University of Augsburg\\
        \texttt{moemke@informatik.uni-augsburg.de}
    \end{center}
\end{minipage}\\
\begin{minipage}{6cm}
\vspace{0.5cm}
    \begin{center}
        \textbf{Émile Naquin}\\
        LaBRI, University of Bordeaux\\
        \texttt{emile.naquin@labri.fr}
    \end{center}
\end{minipage}
}

\begin{document}
	
	\maketitle
	\begin{abstract}
		Robust optimization under interval uncertainty aims to compute solutions that perform well on a range of scenarios that are described by interval-constrained costs. In this paper, we revisit a framework introduced by Ganesh, Maggs and Panigrahi in 2020 to study the robust optimization of NP-hard problems under interval uncertainty. We start by generalizing a result in the $\ell=0$ case, which transforms a category of approximation algorithms into a robust approximation algorithm. Furthermore, in the general case, we provide a theorem that turns any local search-based approximation algorithm into a robust approximation algorithm under three newly formalized conditions over the moves of the local search algorithm. We then use this result to present the first robust approximation algorithm for Weighted $k$-Set Cover, the third NP-hard problem known to admit a robust approximation, and the first since the publication of Ganesh, Maggs and Panigrahi's paper.
	\end{abstract}
    \thispagestyle{empty}
	\newpage
	\setcounter{page}{1}
    
	\section{Introduction}
	
	Real-life optimization problems usually have to deal with uncertainty, such as delays, unpredictable demand, or weather. There is a wide body of literature modeling uncertainty for various kinds of optimization problems. Notably, stochastic models have enjoyed widespread study and use, but they require some information about the probabilities of some possibly extreme events, which is difficult to estimate or model. A somewhat opposite view on the guarantees that we seek is the worst-case model, in which we are interested in ensuring some optimization guarantee on any possible scenario, irrespective of their likeliness.
	
	Consider for example a mailperson who wants to deliver letters to $n$ houses. The mailperson wants to minimize the time it takes them to complete their deliveries, and therefore, this is a Traveling Salesperson Problem (TSP) instance. However, the mailperson cannot know precisely how long it takes to go from one person's home to another, because traffic conditions can evolve, rain can reduce the maximum speed, etc. Therefore, the mailperson will consider that the time needed to go from house $A$ to house $B$ is contained in some interval $[\ell_{AB}, u_{AB}]$. We call a realization of the times between the houses a \textit{scenario}, that is, a vector of the time between each pair of houses where each time is contained in the uncertainty interval for that specific pair of houses.
	
	Now that we have two vectors, $\ell$ and $u$, for the lower- and upper-bounds of the time between pairs of houses, we want to find a way to express that we want a trip that has the best worst-case guarantee. To do so, Ganesh, Maggs and Panigrahi \cite{GaneshRobustAlgorithmsTSP2023} have introduced an approximation framework which ensures that the solution is comparatively close to the optimum on each scenario. For a solution $S$, we look at its \textit{regret} $r_S$, defined as the maximum over all scenarios 
	$w\in[\ell, u]$ of $w(S)-\opt_w$, where $w(S) := \sum_{e \in S} w_e$ and $\opt_w$ is the cost of an optimal solution given the selected vector $w$.
	
	We can now introduce a measure of how uncertain the problem is by letting $\MR=\min_{\text{solution } S} r_S$, the \textit{minimal regret} obtainable by a solution. This is an interesting parameter because it measures how wide the intervals are, but seen through the problem itself. Indeed, for some problems, having an enormous uncertainty on some data points is not a problem (consider for example a spanning tree problem on a graph with a bridge: no matter the weight of the bridge, it will be used in every solution, and therefore its uncertainty will not have any repercussion in $\MR$).
	
	For some $\alpha,\beta\geq1$, we say that a solution $S$ is an {\em $(\alpha, \beta)$-robust approximation} if on every scenario $w\in[\ell, u]$, we have $w(S)\leq \alpha \cdot \opt_w + \beta\cdot\MR$. Thus on every scenario, our solution's weight is comparable to that of an optimal solution to which we add some term $\MR$, which quantifies how uncertain the problem is. Remark that we cannot hope for a one-parameter approximation, as both $\opt_w$ and $\MR$ are necessary here. If we ask for $\beta=0$, this is just a standard approximation algorithm, which does not get more freedom as the uncertainty grows, and thus cannot reach a constant approximation factor in general. On the other hand, if we ask for $\alpha=0$, then in the case where $\ell=u$, there is exactly one scenario, so the minimal regret is $\MR=0$ (this is the regret of the optimal solution on the scenario $\ell=u$), and thus the right-hand side is $0$, so it cannot be satisfied. Furthermore, note that the solution of minimal regret $\MR$ is a $(1,1)$-robust approximation.
	
	\medskip\noindent
	\textbf{Our contributions.} In this article, we make three significant contributions.
	
	(1)  In \Cref{sec:l-zero}, we show how to convert any \emph{difference-approximation algorithm} for a \emph{linear selection} problem (both notions are defined formally below) to a robust approximation algorithm in the case where the lower bound on the weights is $\ell=0$, as stated in the following:
	
	\begin{restatable*}{thm}{thmLZero}
		\label{thm:l-zero}
		Any linear selection problem $\Pp$ that admits a polynomial-time $\gamma$-difference approximation algorithm and a $\lambda$-approximation algorithm also admits a polynomial-time $(\lambda\delta, \lambda(\delta\gamma+1))$-robust approximation algorithm in the $\ell=0$ case, where $\delta$ is an upper bound on the integrality gap of the LP formulation of $\Pp$.
	\end{restatable*}
	
	(2) In the general $\ell$ case, we show that if a local search-based algorithm satisfies three properties that will be defined in \Cref{subsec:prerequisites}, we can turn it into a robust-approximation algorithm using a general theorem. These properties generalize the Steiner Tree robust approximation proof of \cite{GaneshRobustAlgorithmsTSP2023} by abstracting away the specificities of the local search algorithm.
	
	(3) This will allow us to conclude with the first robust approximation for Weighted $k$-Set Cover, letting $H_k$ be the $k$-th harmonic number:
	
	\begin{restatable*}{thm}{thmSetcover}
		\label{thm:setcover}
		The Weighted $k$-Set Cover problem admits a $(O(H_k^4k^2), O(H_k^3k))$-robust approximation.
	\end{restatable*}
	
	 The Weighted $k$-Set Cover problem is defined as follows. Given a set $\Omega$ of sets of cardinality at most $k$, along with a weight function $w:\Omega\rightarrow \RR^+$, we want to compute a cover $S\subs \Omega$ of $\bigcup \Omega$ such that $w(S) = \sum_{s\in S} w(s)$ is minimized. This problem is known to be hard to approximate with a ratio better than $\ln k - O(\ln \ln k)$~\cite{GLLLocalSearchSetCover}.
	
	We believe that this new formalization of properties of local search algorithms that lead to a robust approximation algorithm could lead to new avenues of research in local search-based approximation algorithms. Furthermore, this is the first robust approximation algorithm for an NP-hard problem since Ganesh, Maggs and Panigrahi published their robust approximations for TSP and Steiner Tree in 2020.
	
	\medskip\noindent
	\textbf{Structure of the paper.} After a brief overview of related work and an introduction of the definitions in \Cref{sec:definitions}, we start, in \Cref{sec:l-zero}, with a restricted case in which the lower bound weight vector $\ell$ is known to be $0$, where we show a method to get robust approximation algorithms, that also showcases the method of the general case. Then, in \Cref{sec:general}, we handle the more general case where $\ell \neq 0$. We first formally specify the required prerequisites and then analyze the robust approximation algorithm. Finally, in \Cref{sec:appsc}, we give an application of this generalized framework, for the Weighted $k$-Set Cover problem.
	
	\paragraph*{Related Work}
	\vspace*{-1mm}
	There is a large body of work on stochastic optimization \cite{ShapiroLecturesStochasticProgramming2021}, in which some data about the distribution of the constraints is known, but as mentioned in \cite{KouvelisRobustDiscreteOptimization1997}, this data may not be available, and furthermore, having expected or probable guarantees may not be sufficient if we want to design solutions that work as well as possible no matter the scenario.
	
	Research on optimization under worst-case uncertainty constraints dates back to at least the short paper by Soyster \cite{SoysterTechnicalNote1973}, in which he shows how to convert a linear program (LP) with inexact constraints into an exact auxiliary linear program. The approaches of \cite{FalkTechnicalNoteExact1976} and \cite{SinghConvexProgrammingSetinclusive1982} solve the problem in a similar fashion.
	In \cite{Ben-TalRobustSolutionsUncertain1999}, Ben-Tal and Nemirovski follow a similar idea but allow for ellipsoidal uncertainty, and therefore, the exact auxiliary program they use is a conic quadratic program instead of simply a linear one.
	Another approach is that from \cite{BertsimasPriceRobustness2004}, in which there is a parameter quantifying how many coefficients of the constraint matrix can change inside of their respective uncertainty interval. In each of these papers, it is possible to cast the uncertain linear program into some other program that is also solvable in polynomial time.
	
	In general, it is not possible to transform a program with uncertainty into another program which may be more complex, but that can be solved exactly using standard techniques, and therefore, the idea of min-max regret optimization has surveyed in \cite{AissiMinMaxMin2009} as another approach in which we measure a solution's performance by its difference in weight with the optimal solution on every scenario, and try to minimize the largest difference. In \cite{KasperskiApproximationAlgorithm2006}, Kasperski provides, to the best of our knowledge, the first approximation results for interval uncertainty, by showing how to compute a solution that has a regret that is bounded in terms of the minimal regret a solution can obtain. In a long series of articles summarized in \cite{KasperskiRobustDiscreteOptimization2016}, several complexity results and more approximation algorithms have been found for various well-known problems. In particular, while \cite{AronComplexityRobustSpanning2004} showed that the Robust Minimal Spanning Tree Problem is NP-hard, it admits a $(1,2)$-robust approximation algorithm \cite{KasperskiApproximationAlgorithm2006}. The technique is general and very elegant, as it simply consists in computing the optimal solution on the midpoint scenario (that is, $\frac{\ell+u}2$). However, its obvious limitation is that it can only be used in practice for problems that we can solve in polynomial-time.
	
	Our work mainly builds on the paper by Ganesh, Maggs and Panigrahi \cite{GaneshRobustAlgorithmsTSP2023}, which provides the first approximation results for NP-hard problems (Traveling Salesperson and Steiner Tree) in the min-max regret setting.
	
	\section{Robust linear optimization problems}\label{sec:definitions}
	
	In this article, we consider \emph{linear selection problems}, which we define as problems that can be modeled as 0-1 integer programs (IP) of the form
	\[
	\begin{array}{ll}
		\min & w^Tx \\
		& Ax\leq b\\
		& x \in \{0,1\}^n,
	\end{array}
	\]
	
	where $A$ and $b$ do not depend on $w$. We specify the problem by $\Pp=(n, \Omega, A, b)$, where $\Omega$ is a set of size $n$ such that $x_i$ represents entry $i$ of $\Omega$. A solution $S$ to $\Pp$ is determined by a solution $x$ to the IP, with $x$ the characteristic vector of $S$. In the rest of this article, we will assume that there exists a separation oracle for the $Ax\leq b$ constraints. This will be vital in the analysis of \Cref{alg:general-aso}, as we will use an approximate separation oracle (defined below) for a robust version of this IP, which also has the $Ax\leq b$ constraints.
	
	In order to study the robust variant of this problem, we use a robust formulation of $\Pp$, called the \emph{robust linear optimization problem} $\Pp_R$ which given an instance of $\Pp$ and two vectors $\ell, u \in \RR^n$ such that $\ell_i\leq u_i$ for all $i\in \Omega$ consists in finding
	$$\argmin_{\substack{x\in\{0,1\}^\Omega\\Ax\leq b}}\max_{w\in[\ell, u]} (w^Tx - \opt_w)$$
	
	(we denote $[\ell,u] = \{w\in\RR^n\mid\forall i\in\Omega, w_i\in[\ell_i, u_i]\}$, and $\opt_w$ is the optimal value of $\Pp$ in this instance, over weight function $w$).
	
	Our goal is to transform an approximation algorithm for problem $\Pp$ into a robust approximation algorithm for problem $\Pp_R$. As mentioned in the introduction, we say that a polynomial-time algorithm $\Aa$ is an {\em $(\alpha, \beta)$-robust approximation} for problem $\Pp_R$ if on any instance and for every pair of lower- and upper-bound on the weights $(\ell, u)$, the solution $\SOL$ computed by $\Aa$ is such that: $\forall w\in[\ell,u], w(\SOL) \leq \alpha \cdot \opt_w + \beta \cdot \MR,$ where $\MR$ is the minimal regret achievable by any feasible solution, that is, the optimal value of the robust problem $\Pp_R$.
	
	The core result of \cite{GaneshRobustAlgorithmsTSP2023} is to show that one way to get a robust algorithm is to find an \textit{approximate separation oracle} for the linear programming (LP) relaxation of $\Pp_R$ that follows:
	
	\begin{alignat*}{2}
		\min\quad r \qquad \qquad \qquad &&& \\
		\forall w\in[\ell, u],\quad  w^Tx &\leq \opt_w + r &&\quad \text{(regret constraints)} \\
		Ax&\leq b &&\quad\text{(problem constraints)} \\
		r &\geq 0&& \\
		x&\in{[}0,1{]}^n.&&
	\end{alignat*}
	
	Indeed, a usual separation oracle would not be enough here, because there is an infinite number of constraints (which could be reduced to exponential by looking only at weight functions that take extremal values on every element) that each involve an NP-hard optimum. Therefore, we look at the following relaxed definition:
	
	\begin{defn}[Approximate Separation Oracle]
		A polynomial-time algorithm $\Ss$ is said to be an $(\alpha, \beta)$-\emph{approximate separation oracle (ASO)} for the IP formulation of $\Pp_R$ if, when run on some fractional assignment of the variables $(x, r)$ of the LP, either:
		\begin{itemize}
			\item returns some constraint of the LP formulation that is unsatisfied by $(x,r)$, or
			\item returns \texttt{Feasible}, in which case $x$ is a feasible solution to $\Pp$ and furthermore, the regret constraints of the LP formulation of $\Pp_R$ are approximately satisfied:
			$$\forall w\in[\ell, u], w^Tx\leq \alpha\cdot \opt_w + \beta\cdot r.$$
		\end{itemize}
	\end{defn}
	
	With such an ASO, we can use Theorem 2.1 of \cite{GaneshRobustAlgorithmsTSP2023} to get a robust approximation algorithm:
	
	\begin{thm}[Theorem 2.1 of \cite{GaneshRobustAlgorithmsTSP2023}]\label{th21}
		If:
		\begin{itemize}
			\item there is a polynomial-time $\gamma$-approximation algorithm for $\Pp$,
			\item the integrality gap of the LP formulation of $\Pp$ is at most $\delta$, and
			\item there is a $(\alpha,\beta)$-ASO for the LP formulation of $\Pp_R$,
		\end{itemize}
		then there exists a polynomial-time $(\gamma\delta\alpha, \gamma\delta\beta+\gamma)$-robust approximation algorithm for the robust problem $\Pp_R$.
	\end{thm}
	
	In the context of regret minimization, a natural notion introduced by \cite{GaneshRobustAlgorithmsTSP2023} is that of \emph{difference approximation}, a notion defined on the original problem $\Pp$ which is useful in the analysis of approximation algorithms for $\Pp_R$. Indeed, when we compare some solution $S$ to an optimal solution $\OPT$, the intersection $S\cap\OPT$ does not play a role in the regret, because it gets canceled out since on any scenario $w\in[\ell,u]$, $w(S)-\opt_w = w(S\setminus \OPT_w) - w(\OPT_w\setminus S)$:
	
	\begin{defn}[Difference approximation]
		We say that a solution $S$ is a \emph{$\gamma$-difference approximation} if $w(S\setminus \OPT_w)\leq \gamma w(\OPT_w\setminus S).$
	\end{defn}
	
	Finally, we will also say that, given two values $x, y\geq 0$ and some precision $\eps>0$, $x$ is a \emph{$(1+\eps)$-guess} of $y$ if $x\in[y, (1+\eps)y]$. This will be used in the algorithm to guess the values of some solutions to which we do not have explicit access.
	
	Now that the setting and some important definitions have been introduced, we can start with  the special case where $\ell=0$.
	
	\section{A difference approximation gives a robust approximation in the $\ell=0$ case}\label{sec:l-zero}
	
	In this section, we show the following theorem:
	
	\thmLZero
	
	In order to apply \Cref{th21}, we will show how to get an ASO for the $\ell=0$ case (but general $u$), therefore concluding that if $\Pp$ admits a polynomial-time $\gamma$-difference approximation algorithm $\Aa$ for some $\gamma\geq 1$, then it also admits a robust approximation algorithm.
	
	\noindent
	\begin{algorithm}
		\DontPrintSemicolon
		Check the problem constraints $Ax\leq b$ and return any violated constraint\label{line:check-p-lp}\;
		$w\gets$ cost function such that $\forall i\in\Omega, w_i = u_ix_i$\label{line:weight-vec}\;
		$S\gets$ output of $\Aa$ on $w$\label{line:sprime}\;
		\eIf{$\sum_{i\notin S} u_ix_i > r$\label{line:checkviolated}}{
			\Return{$\sum_{i\notin S} u_ix_i \leq r$ \texttt{is broken}}\label{line:constr-return}
		}{
			\Return{\texttt{Feasible}}
		}
		\caption{$\Pp_R\text{-ASO}(x,r)$: $(1,\gamma)$-approximate separation oracle for the $\Pp_R$ LP.}
		\KwData{$(x,r)$ a tentative solution to the $\Pp_R$ LP.}
		\label{alg:l0-aso}
	\end{algorithm}
	
	\begin{lemma}\label{lemP-ASO}
		Suppose that problem $\Pp$ admits a $\gamma$-difference approximation $\Aa$. Then, for any instance of $\Pp_R$ with weight bounds $\ell=0$ and $u$, there exists an algorithm $\Pp_R$-ASO which is a $(1, \gamma)$-ASO for the LP formulation of $\Pp_R$.
	\end{lemma}
	\begin{proof}
		If there is a violation of a problem constraint, this is directly detected by Line~\ref{line:check-p-lp} of Algorithm~\ref{alg:l0-aso}. If no violation is found, let $S$ be the solution computed on Line~\ref{line:sprime}, and $w$ be the weight vector defined on Line~\ref{line:weight-vec} such that for all $i\in\Omega$, $w_i = u_ix_i$.
		
		The only other constraint violation that the algorithm can output is $\sum_{i\notin S'} u_ix_i \leq r$ on Line~\ref{line:constr-return}. Therefore, we show that if there is some weight function $w'$ and solution $S'$ such that $\sum_{i\in\Omega}w'_ix_i>w'(S')+\gamma r$, the algorithm sees that $\sum_{i\notin S} u_ix_i > r$ on Line~\ref{line:checkviolated}. So, we suppose that such $w', S'$ exist.
		
		We define $w^*$ as $w^*_i=\begin{cases}
			u_i & \text{if } i\notin S' \\
			0 & \text{otherwise}
		\end{cases}$ for all $i\in\Omega$. Then:
		$$w(\Omega\setminus S') = \sum_{i\in \Omega} w^*_ix_i - w^*(S') \geq  \sum_{i\in \Omega} w'_ix_i - w'(S') > \gamma r.$$
		
		Now, remark that $w(\Omega\setminus \OPT_w)\geq w(\Omega\setminus S')>\gamma r$. We can now compare $w(\Omega\setminus \OPT_w)$ to $w(\Omega\setminus S)$ using the fact that $S$ is a $\gamma$-difference approximation:
		\begin{align*}
			w(\Omega\setminus S) &= w(\Omega\setminus (S\cup \OPT_w))+w(\OPT_w\setminus S) \\
			&\geq w(\Omega\setminus (S\cup \OPT_w)) + \frac1\gamma w(S\setminus \OPT_w)\text{ ($\gamma$-difference approximation)} \\
			&\geq \frac{w(\Omega\setminus (S\cup \OPT_w)) + w(S\setminus \OPT_w)}\gamma
			\geq \frac1\gamma w(\Omega\setminus \OPT_w) > r.
		\end{align*}
	\end{proof}
	
	Since \Cref{lemP-ASO} gives an ASO, using \Cref{th21} directly implies \Cref{thm:l-zero}. This result is a simple example of how to use \Cref{th21} and ASOs in order to get robust approximation algorithms in the natural $\ell=0$ case.
	
	\section{The $\ell\neq 0$ case}\label{sec:general}
	
	In this section, we will provide a way to transform a local search approximation algorithm for some problem $\Pp$ to a robust approximation algorithm for the robust version of $\Pp$. The proofs are similar to the ones from the Steiner Tree proof of \cite{GaneshRobustAlgorithmsTSP2023}, but they have been generalized in a way that allows their use on any local search algorithm satisfying three conditions, that will be specified in the next subsection, which will lead to the first robust approximation algorithm for Weighted $k$-Set Cover presented in \Cref{sec:appsc}.
	
	We will assume that there is a local search approximation algorithm $\Ll$ for $\Pp$. We assume that algorithm $\Ll$ works on a potential function $\Phi$ that substitutes the weight function $w$ of the instance. Therefore, each element $i\in\Omega$ gets a potential $\phi(i)$ such that for any solution $S\subseteq \Omega$, $\Phi(S)=\sum_{i\in S} \phi(i)$. If $w=\Phi$, we say that the local search algorithm $\Ll$ is \emph{oblivious}.
	
	Starting from some solution $S\subseteq \Omega$, the algorithm has a set of moves it can make, which is a set $\Pi_S\subseteq 2^\Omega$ of solutions that can be reached in one local move from $S$, and selects the solution $S'\in\Pi_S$ that has the minimal potential. Algorithm $\Ll$ repeats this step until the potential change in one iteration becomes lower than some fixed threshold (this is what usually guarantees a polynomial running time). When we say that $\Ll$ has some $\gamma$-approximation guarantee, we mean that there is a proof that any solution $S$ that is approximately local minimal on $\Phi$ (that is, for all $S'\in\Pi_S$, $\Phi(S)-\Phi(S')\leq\eps$) is a $\gamma$-approximation of the optimal solution on $w$.
	
	\subsection{Prerequisites}\label{subsec:prerequisites}
	We start by specifying formally the three properties \Cref{ppt:bounded-potential}, \Cref{ppt:mmc} and \Cref{efficient-gm} that will be sufficient to get a robust approximation, together with some of their implications.
	
	The {\em well-behaved potential} property asks that the potential can be decomposed as a sum over the elements of $\Omega$, and allows us to transform potentials into weights by losing only a constant factor. This is probably the most constraining property, since it requires a potential that is very closely linked to the weight; figuring out how to relax this property would allow for a wider applicability of this analysis in the $\ell\neq 0$ case.
	
	The {\em minimal move covering} property provides a weight function $\alpha$ over the solutions that can be reached from a particular solution in one move. What this allows us to do is to link the potential of an optimal solution to the potential of the current solution, only through the moves available from the current solution that go ``toward'' the optimal solution (that is, they have some elements present in the optimal solution but not in the current solution). One way to think of this weight function is to take a $0$--$1$ function, which selects a subset of solutions that can be reached from the current solution in one move, in a way that covers the current solution completely and does not overlap too much (to avoid losing some constant by counting some elements too many times). This property implies \Cref{lem510sc}, which guarantees the existence of ``good'' moves relative to two different potential functions at the same time. This will be the fundamental building block of the robust approximation algorithm.
	
	Finally, the {\em efficient good move retrieval} property ensures that a ``good'' enough move can be found in polynomial time, thereby allowing an algorithm to use the result of \Cref{lem510sc}.
	
	\subsubsection{Well-Behaved Potential}
	
	The potential function we use to guide the local search toward a local minimum is a function of the weights of the elements of $\Omega$, and possibly of other properties of the instance.
	
	\begin{ppt}[Well-behaved potential]\label{ppt:bounded-potential}
		The potential function $\Phi\geq 0$ is said to be \emph{well-behaved} if it has the following properties:
		\begin{enumerate}
			\item There exists a function $\phi$ over $\Omega$ such that $\Phi(S)=\sum_{i\in S}\phi(i)$ for all $S\subseteq\Omega$.
			\item There exist two constants $\Lambda_1, \Lambda_2>0$ such that we have
			$\Lambda_1w\leq \Phi\leq \Lambda_2w.$
		\end{enumerate}
	\end{ppt}
	
	Since we will manipulate two potential functions $\Phi$ and $\Phi'$ in the algorithm analysis, we take the convention that they are always associated with functions $\phi, \phi'$ over $\Omega$ such that $\Phi(S)=\sum_{i\in S}\phi(i)$ and $\Phi'(S)=\sum_{i\in S}\phi'(i)$. In the rest of the paper, we set $\Lambda=\Lambda_2/\Lambda_1\geq 1$.
	
	The first item of this property implies the following lemma, proven in \cite{GaneshRobustAlgorithmsTSP2023}:
	
	\begin{lemma}[Lemma 5.12 from \cite{GaneshRobustAlgorithmsTSP2023}]\label{lem512}
		Let $A,B,C$ be three solutions to some instance. Then:
		$$\Phi(A)-\Phi(B) = [\Phi(A\setminus C)-\Phi(C\setminus A)]-[\Phi(B\setminus C)-\Phi(C\setminus B)].$$
	\end{lemma}
	
	We can also use the well-behaved potential property to show the following lemma:
	\begin{lem}\label{lem-ls-dapx}
		Fix any $\gamma\geq 1$, and assuming that the local search algorithm $\Ll$ guarantees a $\gamma$-approximation. Then, any local minimum is also a $\gamma$-difference approximation.
	\end{lem}
	\begin{proof}
		Let $w$ be some weight function over $\Omega$, that is associated with a potential $\Phi$, and $\OPT_w$ an optimal solution on $w$ for problem $\Pp$. Let $S$ be a locally optimal solution for $\Ll$, that is: for all moves $S'\in\Pi_S$, $\Phi(S)-\Phi(S')\leq \eps$. We can create a new weight function $w'$ which is equal to $w$ except on $S\cap\OPT_w$, where it is set to $0$. It is also associated to a potential $\Phi'$. Remark that:
		\[\Phi'(S\cap\OPT_w)\leq \Lambda_2w'(S\cap\OPT_w) = 0.\]
		
		Thus, for all moves $S'\in\Pi_S$, we have $\Phi'(S)-\Phi'(S') \leq\eps$. Therefore, $S$ is also a local minimum on $\Phi'$, and is therefore a $\gamma$-approximation there, so: $w'(S)\leq\gamma w'(\OPT_w)$. By definition of $w'$, this is exactly: $w(S\setminus\OPT_w)\leq\gamma w(\OPT_w\setminus S)$.
	\end{proof}
	
	\subsubsection{Minimal Move Covering}
	
	We now need a property that links the cost of a current solution $S$ and a ``target'' solution $T$, which can be thought of as the optimal solution itself. Recall that we call $\Pi_S$ the set of solutions that we can reach from $S$ in one step of the local search. We look at $\Pi_S^T\subs \Pi_S$ that is defined as the set of moves $S'$ such that $(S'\setminus S)\cap T\neq\es$ and $(S'\setminus S)\subs T$, meaning that performing move $S'$ adds new elements of $T$ that were not before in $S$, and only elements from $T$. The following, somewhat technical property, really makes sense in light of its consequence \Cref{general-convergence}. What this property allows is to compare the potentials $\Phi(S)$ and $\Phi(T)$ through the moves in $\Pi_S^T$. In order for this to be possible, we ask that these moves can put into play every element in $S$ (otherwise, we could not count them; this is the covering property), and that they can introduce each element of $T$ in the next solution (otherwise, we could similarly not account for the weight of elements which are not used by the moves; this is the packing property):
	
	\begin{ppt}[Minimal Move Covering]\label{ppt:mmc}
		Let $d\in\NN$. A local search algorithm $\Ll$ has the {\em $d$-minimal move covering ($d$-MMC) property} if for any two solutions $S$ and $T$ there is a function $\alpha: \Pi_S^T \rightarrow \RR^+$ such that:
		\begin{enumerate}
			\item For all $i\in S\setminus T$, $\sum_{\substack{S'\in\Pi_S^T\\i\in S\setminus S'}} \alpha(S')\geq 1$ (covering of $S$).
			\item For all $i\in T\setminus S$, $\sum_{\substack{S'\in\Pi_S^T\\i\in S'\setminus S}} \alpha(S')\leq d$ (packing on $T$).
			\item $\sum_{S'\in\Pi_S^T} \alpha(S')\leq n$.
		\end{enumerate}
	\end{ppt}
	
	The goal here is to have $d$ as small as possible, since it will play the role of a sort of approximation ratio (apparent in \Cref{general-convergence}). However, for technical reasons that will become apparent in the analysis of the \textsc{GreedySwap} algorithm, we will have to take $d\geq \lambda$, where $\lambda$ is an approximation ratio for problem $\Pp$. This will be used when we compute the initial solution through a $\lambda$-approximation to ensure that it is also a $d$-approximation, thereby initializing the rest of the algorithm.
	
	Surprisingly, the $d$-MMC property implies the following lemma which informally claims that as soon as there is a large proportion $(1/\theta)$ of moves that do not significantly ($\leq\lambda$) improve the solution, we have a difference approximation. We set $\rho(S,S')=\frac{\Phi(S\setminus S')}{\Phi(S'\setminus S)}$ for concision. The reason we look at $\Phi(S\setminus S')$ here and not $\Phi(S'\setminus S)$ is that for a move $S'\in\Pi_S^T$, $\Phi(S\setminus S')$ is the amount of potential of $S$ that will be removed by performing the move $S'$. Therefore, this allows us to talk about the potential of $S$, rather than that of $S'$ (as would be the case with $\Phi(S'\setminus S)$); the comparison with $T$ will come from the $d$-MMC property.
	
	\begin{restatable}{lemma}{lemGenConv}\label{general-convergence}
		Let $S, T$ be solutions to some instance with potential $\Phi$ such that for some $\lambda,\theta\geq 1$, we have:
		$$\sum_{\substack{S'\in\Pi_S^T\\\rho(S,S')\leq\lambda}} \alpha(S')\Phi(S\setminus S') \geq \frac{1}\theta\sum_{S'\in\Pi_S^T}\alpha(S')\Phi(S\setminus S').$$
		
		If the $d$-MMC property is satisfied, then we have: $\Phi(S\setminus T)\leq d\theta\lambda \Phi(T\setminus S)$.
	\end{restatable}
		\begin{proof}
		Assume that $\sum_{\substack{S'\in\Pi_S^T\\\rho(S,S')\leq\lambda}} \alpha(S')\Phi(S\setminus S') \geq \frac{1}\theta\sum_{S'\in\Pi_S^T}\alpha(S')\Phi(S\setminus S')$.
		
		We start by forming the following expression, summing over every move $S'\in\Pi_S^T$ such that $\rho(S,S')\leq\lambda$, and show that this sum is nonnegative:
		
		$$\sigma=\sum_{\substack{S'\in\Pi_S^T\\\rho(S,S')\leq\lambda}}\alpha(S')\Phi(S'\setminus S)-\frac1\lambda\sum_{\substack{S'\in\Pi_S^T\\\rho(S,S')\leq\lambda}}\alpha(S')\Phi(S\setminus S').$$
		
		From the hypothesis, we have:
		\begin{align*}
			\sigma&\leq \sum_{\substack{S'\in\Pi_S^T\\\rho(S,S')\leq\lambda}}\alpha(S')\Phi(S'\setminus S)-\frac1{\theta\lambda}\sum_{\substack{S'\in\Pi_S^T}}\alpha(S')\Phi(S\setminus S') \\
			&\leq \sum_{\substack{S'\in\Pi_S^T}}\alpha(S')\Phi(S'\setminus S)-\frac1{\theta\lambda}\sum_{\substack{S'\in\Pi_S^T}}\alpha(S')\Phi(S\setminus S')\,.
		\end{align*}
		
		And also, since $\sum_{\substack{S'\in\Pi_S^T\\\rho(S,S')\leq\lambda}}\alpha(S')\Phi(S'\setminus S)\geq\frac1\lambda\sum_{\substack{S'\in\Pi_S^T\\\rho(S,S')\leq\lambda}}\alpha(S')\Phi(S\setminus S')$:
		$$\sigma \geq \frac1\lambda\sum_{\substack{S'\in\Pi_S^T\\\rho(S,S')\leq\lambda}}\alpha(S')\Phi(S\setminus S')-\frac1\lambda\sum_{\substack{S'\in\Pi_S^T\\\rho(S,S')\leq\lambda}}\alpha(S')\Phi(S\setminus S') = 0.$$
		
		Therefore, we have:
		$$\sum_{\substack{S'\in\Pi_S^T}}\alpha(S')\Phi(S'\setminus S)
		\geq \frac1{\theta\lambda}\sum_{\substack{S'\in\Pi_S^T}}\alpha(S')\Phi(S\setminus S').$$
		
		Now, using the $d$-MMC property, we can transform the right- and left-hand sides to get the potentials we are interested in. We start with the left-hand side (the first inequality comes from the definition of $\Pi_S^T$):
		\begin{align*}
			\sum_{\substack{S'\in\Pi_S^T}}\alpha(S')\Phi(S'\setminus S) &= \sum_{i\in T\setminus S}\sum_{\substack{S'\in\Pi_S^T\\i\in S'\setminus S}}\alpha(S')\phi(i) \\
			&= \sum_{i\in T\setminus S}\phi(i)\sum_{\substack{S'\in\Pi_S^T\\i\in S'\setminus S}}\alpha(S') \leq d\Phi(T\setminus S) \text{ (by (2) of $d$-MMC)}.
		\end{align*}
		
		Now, with the right-hand side, the argument is similar:
		\begin{align*}
			\frac1{\theta\lambda}\sum_{\substack{S'\in\Pi_S^T}}\alpha(S')\Phi(S\setminus S') &= \frac1{\theta\lambda}\sum_{i\in S}\sum_{\substack{S'\in\Pi_S^T\\i\in S\setminus S'}}\alpha(S')\phi(i) 
			\geq \frac1{\theta\lambda}\sum_{i\in S\setminus T}\phi(i)\sum_{\substack{S'\in\Pi_S^T\\i\in S\setminus S'}}\alpha(S')\\
			&\geq \frac1{\theta\lambda}\Phi(S\setminus T)\quad\text{(by (1) of the $d$-MMC property)}.
		\end{align*}
		
		This concludes: $d\theta\lambda \Phi(T\setminus S)\geq \Phi(S\setminus T)$.
	\end{proof}
	
	This will be used as follows: as long as we do not have a satisfying difference-approximation, there will be many moves that add some parts of the optimal solution which are ``good'' in the sense that they significantly decrease the potential.
	
	This implies the following lemma that will be the version used in the analysis of the algorithm to show that the algorithm makes progress. In the following, we assume that $\Ll$ has the $d$-MMC property.
	
	\begin{restatable}[generalizes Lemma 5.10 from \cite{GaneshRobustAlgorithmsTSP2023}]{lemma}{lemFiveTen}\label{lem510sc}
		Let $S, T$ be solutions to an instance of problem $\Pp$ with two potentials $\Phi$ and $\Phi'$, such that for some $\Gamma>1$, we have $\Phi(S\setminus T)>d\Gamma \Phi(T\setminus S)$. Fix any $\eps\in(0,\sqrt\Gamma-1)$. Then, there exists some move $S'\in \Pi_S^T$ such that (setting $r=S\setminus S'$ and $a=S'\setminus S$):
		$$\frac{(1+\eps)\Phi'(a)-\Phi'(r)}{\Phi(r) - (1+\eps)\Phi(a)} \leq \frac{d(1+\eps)\Gamma}{(\sqrt\Gamma-1)(\sqrt\Gamma-1-\eps)} \frac{\Phi'(T\setminus S)}{\Phi(S\setminus T)}$$
		and $\Phi(r)-(1+\eps)\Phi(a)\geq\frac1{n^2} \Phi(S\setminus T)$.
	\end{restatable}
	\begin{proof}
		With $S$ fixed, we set $r_{S'}=S\setminus S'$ and $a_{S'}=S'\setminus S$.
		
		We follow the averaging argument of \cite{GaneshRobustAlgorithmsTSP2023}. Consider the quantity:
		$$R=\frac{\sum_{\substack{S'\in \Pi_S^T\\ \Phi'(r_{S'})>\sqrt\Gamma \Phi(a_{S'})\\ \Phi(r_{S'})-(1+\eps)\Phi(a_{S'})\geq\frac1{n^2} \Phi(S\setminus T)}} \alpha(S')[(1+\eps)\Phi'(a_{S'})-\Phi'(r_{S'})]}{\sum_{\substack{S'\in \Pi_S^T\\ \Phi(r_{S'})>\sqrt\Gamma \Phi(a_{S'})\\\Phi(r_{S'})-(1+\eps)\Phi(a_{S'})\geq\frac1{n^2} \Phi(S\setminus T)}} \alpha(S')[\Phi(r_{S'})-(1+\eps)\Phi(a_{S'})]},$$
		
		which we want to upper bound. We start by bounding the numerator by $d(1+\eps)\Phi'(T\setminus S)$. Remark that:
		\begin{align*}
			&\sum_{\substack{S'\in \Pi_S^T\\ \Phi'(r_{S'})>\sqrt\Gamma \Phi(a_{S'})\\ \Phi(r_{S'})-(1+\eps)\Phi(a_{S'})\geq\frac1{n^2} \Phi(S\setminus T)}} \alpha(S')[(1+\eps)\Phi'(a_{S'})-\Phi'(r_{S'})]\\
			\qquad&\leq (1+\eps)\sum_{\substack{S'\in \Pi_S^T\\ \Phi(r_{S'})>\sqrt\Gamma \Phi(a_{S'})\\ \Phi(r_{S'})-(1+\eps)\Phi(a_{S'})\geq\frac1{n^2} \Phi(S\setminus T)}} \alpha(S')\Phi'(a_{S'}) \\
			&= (1+\eps)\sum_{i\in T}\phi'(i) \sum_{\substack{S'\in \Pi_S^T\\ \Phi(r_{S'})>\sqrt\Gamma \Phi(a_{S'})\\ \Phi(r_{S'})-(1+\eps)\Phi(a_{S'})\geq\frac1{n^2} \Phi'(S\setminus T)\\i\in a_{S'}}} \alpha(S') \\
			&= (1+\eps)\sum_{i\in T\setminus S}\phi'(i) \sum_{\substack{S'\in \Pi_S^T\\ \Phi(r_{S'})>\sqrt\Gamma \Phi(a_{S'})\\ \Phi(r_{S'})-(1+\eps)\Phi(a_{S'})\geq\frac1{n^2} \Phi'(S\setminus T)\\i\in a_{S'}}} \alpha(S') \\
			&\leq (1+\eps)d\Phi'(T\setminus S).
		\end{align*}
		
		The second to last equality comes from the fact that by definition, for every move $S'$, we have $a_{S'}\cap S=\es$. This proves the claimed upper bound.
		
		Now, we lower bound the denominator. Observe that:
		
		\begin{ceqn}
			\begin{multline*}
				\sum_{\substack{S'\in \Pi_S^T\\ \Phi(r_{S'})>\sqrt\Gamma \Phi(a_{S'})\\\Phi(r_{S'})-(1+\eps)\Phi(a_{S'})\geq\frac1{n^2} \Phi(S\setminus T)}} \alpha(S')[\Phi(r_{S'})-(1+\eps)\Phi(a_{S'})]\\
				\geq \sum_{\substack{S'\in \Pi_S^T\\ \Phi(r_{S'})>\sqrt\Gamma \Phi(a_{S'})}} \alpha(S')[\Phi(r_{S'})-(1+\eps)\Phi(a_{S'})] \\
				- \sum_{\substack{S'\in \Pi_S^T\\\Phi(r_{S'})-(1+\eps)\Phi(a_{S'})<\frac1{n^2} \Phi(S\setminus T)}} \alpha(S')[\Phi(r_{S'})-(1+\eps)\Phi(a_{S'})].
			\end{multline*}
		\end{ceqn}

		We start by upper bounding the second term of the right-hand side of the equation above:
		
		\begin{align*}
			&\sum_{\substack{S'\in \Pi_S^T\\\Phi(r_{S'})-(1+\eps)\Phi(a_{S'})<\frac1{n^2} \Phi(S\setminus T)}} \alpha(S')[\Phi(r_{S'})-(1+\eps)\Phi(a_{S'})]\\
			\qquad &\leq \frac1{n^2}\Phi(S\setminus T)\sum_{\substack{S'\in \Pi_S^T\\\Phi(r_{S'})-(1+\eps)\Phi(a_{S'})<\frac1{n^2} \Phi(S\setminus T)}} \alpha(S')\\
			\qquad &\leq \frac{\Phi(S\setminus T)}{n}\text{ by condition (3) of the $d$-MMC property}.
		\end{align*}

		We will now use \Cref{general-convergence} (with $\lambda=\theta=\sqrt\Gamma$) to lower bound the first term:
		
		\begin{align*}
			 &\sum_{\substack{S'\in \Pi_S^T\\ \Phi(r_{S'})>\sqrt\Gamma \Phi(a_{S'})}} \alpha(S')[\Phi(r_{S'})-(1+\eps)\Phi(a_{S'})]\\
			\qquad&\geq \sum_{\substack{S'\in \Pi_S^T\\ \Phi(r_{S'})>\sqrt\Gamma \Phi(a_{S'})}} \alpha(S')[\Phi(r_{S'})-\frac{(1+\eps)}{\sqrt\Gamma}\Phi(r_{S'})]\\
			&= \frac{\sqrt\Gamma - 1 - \eps}{\sqrt\Gamma}\sum_{\substack{S'\in \Pi_S^T\\ \Phi(r_{S'})>\sqrt\Gamma \Phi(a_{S'})}} \alpha(S')\Phi(r_{S'}) \\
			&\geq \frac{(\sqrt \Gamma-1)(\sqrt\Gamma - 1 - \eps)}{\Gamma}\sum_{\substack{S'\in \Pi_S^T}} \alpha(S')\Phi(r_{S'}).
		\end{align*}
		
		We show that $\sum_{S'\in\Pi_S^T} \alpha(S')\Phi(r_{S'})\geq\Phi(S\setminus T)$. To do so, remark that:
		\[
		\sum_{S'\in \Pi_S^T} \alpha(S')\Phi(r_{S'}) = \sum_{i\in S} \phi(i) \sum_{\substack{S' \in \Pi_S^T\\ i\in r_{S'}}} \alpha(S')
		\geq \sum_{i\in S} \phi(i) 
		= \Phi(S) \geq \Phi(S\setminus T).
		\]
		
		Therefore, we conclude that:
		
		$$R\leq \frac{(1+\eps)d}{\frac{(\sqrt \Gamma-1)(\sqrt\Gamma - 1 - \eps)}{\Gamma}-\frac1n}\frac{\Phi'(T\setminus S)}{\Phi(S\setminus T)},$$
		
		and the result follows for $n$ large enough.
	\end{proof}
	
	\subsubsection{Efficient Good Move Retrieval}
	
	We want to be able to find a move that is good enough, {\em i.e.}, that is close to being as good (for both potentials $\Phi$ and $\Phi'$) as the one provided by the previous lemma.
	
	\begin{ppt}[generalizes Lemma 5.11 of \cite{GaneshRobustAlgorithmsTSP2023}]\label{efficient-gm}
		Let $\eps>0$. Given two potentials $\Phi$ and $\Phi'$, a solution $S$ and some bounds $W,W' > 0$, assume there is a move $S'\in\Pi_S$ such that $\Phi(S'\setminus S)-\Phi(S\setminus S')\leq W$ and $\Phi'(S'\setminus S)-\Phi'(S\setminus S')\leq W'$. Then, there is a polynomial-time algorithm that finds a move $S''$ such that:
		\begin{itemize}
			\item $\Phi(S''\setminus S) - \Phi(S\setminus S'')\leq(1+\eps)W$ and
			\item $\Phi'(S''\setminus S)-\Phi'(S\setminus S'')\leq (1+\eps)W'$.
		\end{itemize}
	\end{ppt}
	
	\subsection{The robust approximation algorithm}
	
	We show that a local search that satisfies the three required properties \Cref{ppt:bounded-potential}, \Cref{ppt:mmc} and \Cref{efficient-gm} of the previous subsection can be used to get a robust approximation algorithm by providing an ASO for the LP formulation of $\Pp_R$ and concluding using \Cref{th21}.
	
	We follow the idea of alternate minimization from \cite{GaneshRobustAlgorithmsTSP2023}, which optimizes a solution over two different weight functions by alternating between a forward phase and a backward phase, both concentrated on optimizing one weight function while ensuring that the other does not get out of control. For both phases, the \textsc{GreedySwap} algorithm~\ref{fig:greedyswap} extracts a ``good move'' using the efficient good move retrieval property of the local search. The \textsc{DoubleApprox} algorithm~\ref{fig:doubleapprox} is essentially the same as in \cite{GaneshRobustAlgorithmsTSP2023}, with potentials instead of weight functions, and some constants changed to $d$ and $\Lambda$. However, the \textsc{GreedySwap} algorithm~\ref{fig:greedyswap} is changed to be more general, and relies on \Cref{efficient-gm} to actually deal with problem $\Pp$.
	
	\begin{figure}[t]
		\begin{algorithm}[H]
			\SetNoFillComment
			\DontPrintSemicolon
			\KwData{A set $\Omega$ of items and potentials $\Phi, \Phi'$ for which all $\phi'(i)$ are a multiple of $\frac{\eps}{n}\chi'$, $\chi>0$ such that $\chi \in [\Phi(\sol), \sqrt{1+\eps}\cdot \Phi(\sol)]$, $\chi'>0$ such that $\chi' \in [\Phi'(\sol), \sqrt{1+\eps}\cdot \Phi'(\sol)]$, constants $\Gamma, \Gamma', \kappa,d,\eps$.}
			$i \leftarrow 0$\;
			\tcc{The following constants are used to guarantee polynomial-time termination}
			$\zeta' = \frac{d(1+\eps)\Gamma'}{(\sqrt{\Gamma'}-1)(\sqrt{\Gamma'}-1-\eps)(d\Gamma'-1)(d\Gamma-1)}$\;
			$\eta=\left[\min\left\{\frac{d\Gamma - 1}{2d\Gamma} ,\frac{(d\Gamma - 1)(\sqrt{\Gamma}-1)(\sqrt{\Gamma}-1-\eps)\kappa}{d^2(1+\eps)\Gamma^2}\right\} - (e^{\zeta' (d\Gamma'+\kappa+1+\eps)} - 1)\right]/\left[1 + \frac{d\Gamma}{d\Gamma - 1} (e^{\zeta' (d\Gamma'+\kappa+1+\eps)} - 1)\right]$\;
			$I = 2(\lceil  \log \frac{n \max_{i\in\Omega} \phi(i)}{\min_{\substack{i\in\Omega\\\phi(i)>0}} \phi(i)} / \log (1+\eta) \rceil+1)$\;
			$\alg^{(0)} \leftarrow $ $\lambda$-approximation of optimal solution with respect to $\Phi'$ for $\lambda< d\Gamma'$ \label{line:approx}\;
			\While{$i\leq I$}{\label{line:mainwhile} \tcc{Iterate over all guesses $\rho$ for $\Phi(\alg^{(i)} \setminus \sol)$}
				$m\leftarrow \min_{\substack{i\in\Omega\\\phi(i)>0}} \phi(i)$\;
				\For{$\rho \in \{m, (1+\eps)m, \ldots (1+\eps)^{\lceil{\log_{1+\eps} \frac{n}{m}\max_{i\in\Omega}\phi(i)\rceil}} m\}$}{
					$\alg^{(i+1)}_\rho \leftarrow \alg^{(i)}$\;
					\tcc{Forward phase}
					\While{$\Phi'(\alg^{(i+1)}_\rho) \leq (d\Gamma'+\kappa)\chi'$ and $\Phi(\alg^{(i+1)}_\rho) > \Phi(\alg^{(i)}) - \rho/2$}{\label{line:firstwhile} 
						$(\alg^{(i+1)}_\rho, stop)\leftarrow \textsc{GreedySwap}(\alg^{(i+1)}_\rho, \Phi, \Phi', \chi, \chi', \frac{1}{10n^2}\rho)$\;
						\If{$stop = 1$}{\label{line:stopcheckfirst}\
							\textbf{break} while loop starting on Line~\ref{line:firstwhile}\;
						}\label{line:stopchecklast}
					}
					$\alg^{(i+2)}_\rho \leftarrow \alg^{(i+1)}_\rho$\;
					\tcc{Backward phase}
					\While{$\Phi'(\alg^{(i+2)}_\rho) \geq d\Gamma'\chi'$}{\label{line:secondwhile} 
						\label{line:secondswap} $(\alg^{(i+2)}_\rho, \sim) \leftarrow \textsc{GreedySwap}(\alg^{(i+2)}_\rho, \Phi', \Phi, \chi', \chi, \frac{\eps}{n}\chi')$\;
					}
				}
				\tcc{If $\Phi(\alg^{(i)}\setminus\sol)=0$, then $\alg^{(i)}$ satisfies \Cref{lem:55}}
				$\alg_0^{i+2}\leftarrow \alg^{(i)}$\;
				$\alg^{(i+2)} \leftarrow \argmin_{\rho} \Phi(\alg^{(i+2)}_\rho)$\;
				$i \leftarrow i + 2$\;
			}
			\textbf{return} all values of $\alg^{(i)}_\rho$ stored for any value of $i, \rho$\;
			\caption{\textsc{DoubleApprox}$(\Omega, \Phi, \Phi', \chi, \chi', \Gamma, \Gamma', \kappa)$ from \cite{GaneshRobustAlgorithmsTSP2023} adapted to our notation, uses the local search algorithm $\Ll$ to compute a solution that performs well on both $\Phi$ and $\Phi'$.}
			\label{fig:doubleapprox}
		\end{algorithm}
	\end{figure}
	
	\begin{figure}[t]
		\begin{algorithm}[H]
			\DontPrintSemicolon
			\KwData{Solution $\alg$, potential functions $\Phi, \Phi'$, $\chi \in [\Phi(\sol), \sqrt{1+\eps}\Phi(\sol)]$ with $\chi>0$, $\chi' \in [\Phi'(\sol), \sqrt{1+\eps}\Phi'(\sol)]$ with $\chi'>0$, minimum improvement per move $\rho$}
			$\Tt \leftarrow \emptyset$\;
			$m\leftarrow \min_{\substack{i\in\Omega\\\phi(i)>0}}\phi(i)$\;
			$m'\leftarrow \min_{\substack{i\in\Omega\\\phi'(i)>0}}\phi'(i)$\;
			\For{$W \in \{m, \sqrt{1+\eps}m, \ldots, \sqrt{1+\eps}^{\lceil \log_{1 + \eps} \chi \rceil}m\}$}{
				\For{$W' \in \{m', \sqrt{1+\eps}m', \ldots, \sqrt{1+\eps}^{\lceil \log_{\sqrt{1 + \eps}} \chi' \rceil}m'\}$}{
					Using \Cref{efficient-gm}, try to find a move $S'\in\Pi_\alg$ such that $\Phi(S'\setminus \alg)-\Phi(\alg\setminus S')\leq \sqrt{1+\eps}W$ and $\Phi'(S'\setminus \alg)-\Phi'(\alg\setminus S')\leq \sqrt{1+\eps}W'$\label{line:increaseW}, if it exists\;
					\If{move $S'$ exists and $\Phi(\alg\setminus S') - \Phi(S'\setminus\alg)\geq \rho$\label{line:move-improves-enough}}{
						$\Tt\leftarrow \Tt\cup \{S'\}$\;
					}
				}
			}
			\If{$\Tt = \emptyset$}{
				\textbf{return} $(\alg, 1)$\;
			}
			$S^* \leftarrow \argmin_{S'\in\Tt} \frac{\Phi'(S\setminus\alg)-\Phi'(\alg\setminus S)}{\Phi(\alg\setminus S)-\Phi(S\setminus\alg)}$\;\label{line:argmin-da}
			\textbf{return} $(S^*, 0)$\;
			\caption{\textsc{GreedySwap}$(\alg, \Phi, \Phi', \chi, \chi', \rho)$, which finds a move which approximates the properties described in Lemma~\ref{lem510sc} using \Cref{efficient-gm}.}
			\label{fig:greedyswap}
		\end{algorithm}
	\end{figure}
	
	\subsection{Analysis of the algorithm}
	
	We can now turn to proving the main theorem of this paper:
	\begin{thm}[generalizes Theorem 5.1 of \cite{GaneshRobustAlgorithmsTSP2023}]\label{thm:main}
		Let $\Pp$ be a linear selection problem having a local search-based approximation algorithm $\Ll$ that has a well-behaved potential with parameter $\Lambda$, the $d$-minimal covering property, and efficient good move retrieval. Suppose that $\Pp$ also admits a $\lambda$-approximation algorithm, possibly not local-search based, and that the integrality gap of the LP formulation of $\Pp$ is at most $\delta$.
		
		Then, the robust version of problem $\Pp$ admits a ${(17\lambda\delta d\Lambda[2\Lambda(9d+1) + 1] + \lambda\delta, \lambda[17\delta d\Lambda+1])}$-robust approximation.
	\end{thm}

    This is the most technical section of the paper. Its goal is to rewrite the proofs used in \cite{GaneshRobustAlgorithmsTSP2023} in a way that makes them more general. We therefore also use their idea of a double-approximation algorithm \textsc{DoubleApprox} (\Cref{fig:doubleapprox}). The analysis of this algorithm necessitates a number of very technical lemmas, so we will provide an overview of the way they work together. Some proofs are deferred to the appendix when they are too long and very similar to theones from \cite{GaneshRobustAlgorithmsTSP2023}.
    
    In this section, we will consider that we are in the $(\alpha,\beta)$-approximate separation oracle (the values of $\alpha$ and $\beta$ will be set later), and that we need to approximately separate some valuation of the variables $(x,r)$. To do so, we assume that there exists some solution $\sol$ that breaks an approximate constraint, that is, there exists some weight function $w_\sol\in[\ell,u]$ such that $w_\sol(x)>\alpha w_\sol(\sol)+\beta r$. This solution $\sol$ and the weight function $w_\sol$ are not known; we want to show that if they do indeed exist, then we can find some constraint violation. \Cref{lem:55} is the key to this approach: if such a solution $\sol$ exists, it ensures that we can find in polynomial time another solution $\alg_j$ which is close enough to $\sol$ to break a constraint, that is: there exists some $w_{\alg_j}\in[\ell,u]$ such that $w_{\alg_j}(x)>w_{\alg_j}(\alg_j)+r$.

    Our goal is to get to \Cref{lem:55}, which is stated below, and shows that algorithm \textsc{DoubleApprox} returns a polynomial-size set of solutions, one of which satisfies two approximation constraints on two weight vectors $w$ and $w'$. This will give a solution that is controlled enough to imply \Cref{lemFiveSix}, which shows that if there is a broken constraint, there will necessarily be a witness of this among this polynomial-size set of solutions.
    
    Now, let's give some intuition about the way \textsc{DoubleApprox} works. A detailed overview can be found in \cite{GaneshRobustAlgorithmsTSP2023} (Section 5.2); we provide a summary of the idea in the following paragraph.
    	
	What the \textsc{DoubleApprox} algorithm does is that assuming two $(1+\eps)$-precise guesses $\chi$ of $\Phi(\sol)$ and $\chi'$ of $\Phi'(\sol)$, the algorithm improves a current solution $\ALG^{(i)}$ by first guessing $\Phi(\ALG^{(i)}\setminus \sol)$ up to $(1+\eps)$ precision, and then improving it over the potential $\Phi$ in a forward phase and subsequently improving it over the potential $\Phi'$ in a backward phase, each phase ending when the solution is good enough over the potential currently being optimized or when the other potential gets too large (in the forward phase). These bounds are expressed as functions of $\chi$ and $\chi'$. Both phases call the \textsc{GreedySwap} algorithm (\Cref{fig:greedyswap}) which uses the Efficient Good Move Retrieval property (\Cref{efficient-gm}) to find local search moves that improve the solution on the current potential while ensuring a bounded loss on the other potential.
	
	\begin{req}
		The case where $\Phi(\sol)=0$ or $\Phi'(\sol)=0$ cannot be handled satisfyingly by \textsc{DoubleApprox}. We analyze this case separately in the proof of \Cref{lem:55}. Therefore, in the algorithm itself, we can always assume that $\max_{i\in\Omega}\phi(i)>0$.
	\end{req}
	
	\begin{restatable}[generalizes Lemma 5.5 of \cite{GaneshRobustAlgorithmsTSP2023}]{lemma}{lemFiveFive}\label{lem:55}
		Let $\sol$ be a solution to problem $\Pp$, and $(x,r)$ some fractional solution of the robust LP, such that $x$ is a feasible solution for the non-robust LP relaxation. Let $w$ and $w'$ be two weight functions such that $w_i=u_ix_i+\ell_i(1 - x_i)$ and $w'_i = \ell_i(1-x_i)$. Let $\Gamma'>1, \kappa>0, 0<\eps<1/4$. Then, there exists a constant $\Gamma$ and a polynomial-time algorithm that obtains a collection of solutions $\ALG$, one of which (say $\ALG_j$) satisfies:
		\begin{itemize}
			\item $w(\ALG_j\setminus \sol)\leq d\Gamma\Lambda\cdot w(\sol\setminus \ALG_j)$ and
			\item $w'(\ALG_j)\leq (d\Gamma'+\kappa+1+\eps)\Lambda w'(\sol)$.
		\end{itemize}
	\end{restatable}
	\vspace{0.3cm}
	
	In the rest of this section, we define $w, w'$ as in \Cref{lem:55}, and let $\Phi$ (resp. $\Phi'$) be the potential associated to $w$ (resp. $w'$), associated with the function $\phi$ (resp. $\phi'$) over $\Omega$ representing the potential of each element of $\Omega$ according to $\Phi$ (resp. $\Phi'$).
	
	First, \Cref{lemma:forwardphase} and \Cref{lemma:backwardphase} provide bounds on the progress of the forward and backward phase respectively. The values of the parameters $\Gamma, \Gamma',\kappa$ are set in \Cref{thm:main}.
	
	\begin{restatable}[Forward Phase Analysis, generalizes Lemma 5.14 of \cite{GaneshRobustAlgorithmsTSP2023}]{lemma}{lemFiveFourteen}\label{lemma:forwardphase}
		For any even $i$ in algorithm \textsc{DoubleApprox}, let:
		$$\rho=(1+\eps)^{\left\lceil \log_{1+\eps}\left(\frac{\Phi(\alg^{(i)}\setminus\sol)}{\min_{\substack{i\in\Omega\\\phi(i)>0}}\phi(i)}\right)\right\rceil}\min_{\substack{i\in\Omega\\\phi(i)>0}}\phi(i)$$
		be such that $\rho \in [\Phi(\alg^{(i)} \setminus \sol),(1+\eps)\Phi(\alg^{(i)} \setminus \sol) ]$.
		Suppose all values of $\alg^{(i+1)}_\rho$ and the final value of $\alg^{(i)}$ in \textsc{DoubleApprox} satisfy $\Phi(\alg^{(i+1)}_\rho \setminus \sol) > d\,\Gamma \cdot \Phi(\sol \setminus \alg^{(i+1)}_{\rho})$ and $\Phi(\alg^{(i)} \setminus \sol) > d\,\Gamma \cdot \Phi(\sol \setminus \alg^{(i)})$.
		Then for $0 < \eps < 2/3 - \frac5{3d\Gamma}$, the final values of $\alg^{(i)}, \alg^{(i+1)}_\rho$ satisfy
		$$\Phi(\alg^{(i)}) - \Phi(\alg^{(i+1)}_\rho) \geq \min\left\{\frac{d\Gamma - 1}{2d\Gamma} ,\frac{(d\Gamma - 1)(\sqrt{\Gamma}-1)(\sqrt{\Gamma}-1-\eps)\kappa}{d^2(1+\eps)\Gamma^2}\right\} \cdot \Phi(\alg^{(i+1)}_\rho \setminus \sol).$$
	\end{restatable}
	
	\begin{restatable}[Backward Phase Analysis, generalizes Lemma 5.15 of \cite{GaneshRobustAlgorithmsTSP2023}]{lemma}{lemFiveFifteen}\label{lemma:backwardphase}
		Fix any even $i+2$ in algorithm \textsc{DoubleApprox} and any value of $\rho$. Suppose all values of $\alg^{(i+2)}_{\rho}$ satisfy $\Phi(\alg^{(i+2)}_{\rho} \setminus \sol) > d\,\Gamma \cdot \Phi(\sol \setminus \alg^{(i+2)}_{\rho})$. Let $T = \frac{\Phi'(\alg^{(i+1)}_{\rho}) - \Phi'(\alg^{(i+2)}_{\rho})}{\Phi'(\sol)}$. Then for 
		$\zeta' = \frac{d(1+\eps)\Gamma'}{(\sqrt{\Gamma'}-1)(\sqrt{\Gamma'}-1-\eps)(d\Gamma'-1)(d\Gamma-1)}, $
		the final values of $\alg^{(i+1)}_\rho, \alg^{(i+2)}_\rho$ satisfy
		$$\Phi(\alg^{(i+2)}_\rho) - \Phi(\alg^{(i+1)}_\rho) \leq (e^{\zeta' T} - 1) \cdot \Phi(\alg^{(i+1)}_\rho \setminus \sol).$$
	\end{restatable}
	
	The following \Cref{cor:swapgain}, which analyzes one full forward-backward iteration, follows from the two previous lemmas and shows that one iteration improves the potential of the solution currently being developed.
	
	\begin{restatable}[generalizes Corollary 5.16 of \cite{GaneshRobustAlgorithmsTSP2023}]{cor}{corFiveSixteen}\label{cor:swapgain}
		Fix any positive even value of $i+2$ in algorithm \textsc{DoubleApprox}, and let $\rho$ be as in \Cref{lemma:backwardphase}, such that $\rho \in [\Phi(\alg^{(i)} \setminus \sol),(1+\eps)\Phi(\alg^{(i)} \setminus \sol) ]$.
		Suppose all values of $\alg^{(i+1)}_\rho$ and the final value of $\alg^{(i)}$ in \textsc{DoubleApprox} satisfy $\Phi(\alg^{(i+1)}_\rho \setminus \sol) > d\Gamma \cdot \Phi(\sol \setminus \alg^{(i+1)}_{\rho})$ and $\Phi(\alg^{(i)} \setminus \sol) > d\Gamma \cdot \Phi(\sol \setminus \alg^{(i)})$. Then for $0 < \eps < \frac23 - \frac5{3d\Gamma}$ and $\zeta'$ as defined in Lemma~\ref{lemma:backwardphase}, the final values of $\alg^{(i+2)}, \alg^{(i)}$ satisfy
		$$\Phi(\alg^{(i)}) - \Phi(\alg^{(i+2)}) \geq$$
		$$\left[\min\left\{\frac{d\Gamma - 1}{2d\Gamma} ,\frac{(d\Gamma - 1)(\sqrt{\Gamma}-1)(\sqrt{\Gamma}-1-\eps)\kappa}{d^2(1+\eps)\Gamma^2}\right\} - (e^{\zeta' (d\Gamma'+\kappa+1+\eps)} - 1)\right] \cdot \,\Phi(\alg^{(i+1)}_\rho \setminus \sol)\,.$$
	\end{restatable}
	\begin{proof}
		We can lower-bound $\Phi(\alg^{(i+1)})-\Phi(\alg_\rho^{(i+2)})$, which is not greater than $\Phi(\alg^{(i+1)})-\Phi(\alg^{(i+2)})$ because of Line~\ref{line:argmin-da} of \textsc{DoubleApprox}. First, we have:
		$$\Phi'(\alg^{(i+1)}_\rho) - \Phi'(\alg^{(i+2)}_\rho) \leq \Phi'(\alg^{(i+1)}_\rho) \leq (d\Gamma' + \kappa + 1 + \eps)\Phi'(\sol),$$
		
		because when the loop in Line~\ref{line:firstwhile} concludes, we may have a violation of the condition $\Phi'(\alg_\rho^{(i+1)})\leq (d\Gamma'+\kappa)\Phi'(\sol)$, however no move can increase $\Phi'(\alg_\rho^{(i+1)})$ by more than $(1+\eps)\Phi'(\sol)$ because of Line~\ref{line:increaseW}, and therefore the violation is at most $(1+\eps)\Phi'(\sol)$.
		
		Then we apply Lemma~\ref{lemma:forwardphase} to $\Phi(\alg^{(i)}) - \Phi(\alg^{(i+1)}_\rho)$ and Lemma~\ref{lemma:backwardphase} to $\Phi(\alg^{(i+1)}_\rho) - \Phi(\alg^{(i+2)}_\rho)$ (using the bound $T \leq d\Gamma' + \kappa + 1 +\eps$ that we have just shown) to get:
		
		\begin{eqnarray*}
			&&\hspace*{-8mm} \Phi(\alg^{(i)}) - \Phi(\alg^{(i+2)}_\rho) = [\Phi(\alg^{(i)}) - \Phi(\alg^{(i+1)}_\rho)] + [\Phi(\alg^{(i+1)}_\rho) - \Phi(\alg^{(i+2)}_\rho)]\\
			&\geq& \left[\min\left\{\frac{d\Gamma - 1}{2d\Gamma} ,\frac{(d\Gamma - 1)(\sqrt{\Gamma}-1)(\sqrt{\Gamma}-1-\eps)\kappa}{d^2(1+\eps)\Gamma^2}\right\} - (e^{\zeta' (d\Gamma'+\kappa+1+\eps)} - 1)\right]\\
			&\qquad& \quad \cdot \,\Phi(\alg^{(i+1)}_\rho \setminus \sol).
		\end{eqnarray*}
	\end{proof}
	
	Finally, the following lemma, whose proof is left to the appendix, concludes about the fact that the \textsc{DoubleApprox} algorithm finds a solution $\alg^*$  that has two approximation properties on $\Phi$ and $\Phi'$.
	
	\begin{restatable}[generalizes Lemma 5.17 from \cite{GaneshRobustAlgorithmsTSP2023}]{lemma}{lemFiveSeventeen}\label{lemma:algconverges}
		Suppose $\Gamma, \Gamma', \kappa$, and $\eps$ are chosen such that for $\zeta'$ as defined in Lemma~\ref{lemma:backwardphase}, 
		$$\min\left\{\frac{d\Gamma - 1}{2d\Gamma} ,\frac{(d\Gamma - 1)(\sqrt{\Gamma}-1)(\sqrt{\Gamma}-1-\eps)\kappa}{d^2(1+\eps)\Gamma^2}\right\} - (e^{\zeta' (d\Gamma'+\kappa+1+\eps)} - 1) > 0\,,$$
		and $0 < \eps < \frac23 -\frac5{3d\Gamma}$. Let $\eta$ equal
		$$\frac{\min\left\{\frac{d\Gamma - 1}{2d\Gamma} ,\frac{(d\Gamma - 1)(\sqrt{\Gamma}-1)(\sqrt{\Gamma}-1-\eps)\kappa}{d^2(1+\eps)\Gamma^2}\right\} - (e^{\zeta' (d\Gamma'+\kappa+1+\eps)} - 1)}{1 + \frac{d\Gamma}{d\Gamma - 1} (e^{\zeta' (d\Gamma'+\kappa+1+\eps)} - 1)}\,.$$
		Assume $\eta > 0$ and let $I = 2(\lceil  \log \frac{n \max_{i\in\Omega} \phi(i)}{\min_{\substack{i\in\Omega\\\phi(i)>0}} \phi(i)} / \log (1+\eta) \rceil+1)$. Then there exists some intermediate value $\alg^*$ assigned to $\alg^{(i)}_\rho$ by the \textsc{DoubleApprox} algorithm for some $i \leq I$ and $\rho$ such that $\Phi(\alg^* \setminus \sol) \leq d\Gamma \Phi(\sol \setminus \alg^*)$ and $\Phi'(\alg^*) \leq (d\Gamma' +\kappa+ 1 + \eps) \Phi'(\sol)$.
	\end{restatable}
	\vspace{0.3cm}
	
	\begin{algorithm}[t]
		\DontPrintSemicolon
		\KwData{$(x,r)$ a tentative solution to the $\Pp_R$ LP.}
		Check the problem constraints $Ax\leq b$ and return any violated constraint\label{line:aso-check-p-lp}\;
		$w\gets$ weight function such that $\forall i\in\Omega, w_i = u_ix_i-\ell_i(x_i-1)$\;
		$w'\gets$ weight function such that $\forall i\in\Omega, w'_i = \ell_i(1-x_i)$\;
		$\alg\gets$ output of \Cref{lem:55} on $w, w'$\label{line:aso-sprime}\;
		\For{$\alg_j\in\alg$}{
			\If{$\sum_{e\notin \alg_j}u_ex_e+\sum_{e\in \alg_j}\ell_e(x_e-1)>r$\label{line:aso-checkviolated}}{
				$\bar w\gets$ weight function such that $\forall i\in\Omega$, $w_i=\begin{cases}
					\ell_i & \text{if } i\in \alg_j \\
					u_i & \text{otherwise}
				\end{cases}$\;
				\Return{$\bar w(x) \leq \bar w(\alg_j)+r$\texttt{ is broken}}\label{line:aso-constr-return}
			}
		}
		\Return{\texttt{Feasible}}
		\caption{\textsc{General}-$\Pp_R$-\textsc{aso}$(x, r)$: Approximate separation oracle for the $\Pp_R$ LP.}
		\label{alg:general-aso}
	\end{algorithm}
	
	These results imply \Cref{lem:55}, which provides a very controlled solution on $w$ and $w'$.
	
	Now, we get to the last lemma before \Cref{thm:main}, which provides an ASO described in Algorithm~\ref{alg:general-aso}, and implies \Cref{thm:main} through \Cref{th21}.
	
	\begin{restatable}[generalizes Lemma 5.6 of \cite{GaneshRobustAlgorithmsTSP2023}]{lemma}{lemFiveSix}\label{lemFiveSix}
		Fix any $\Gamma'>1, \kappa>0, \eps\in(0,1/4)$, and let $\Gamma$ be as in \Cref{lem:55}. Let $\alpha=d\Gamma\Lambda(\Lambda[d\Gamma'+\kappa+1+\eps]+1)+1$ and $\beta=d\Gamma\Lambda$. Then, there exists an $(\alpha,\beta)$-ASO for the LP formulation of the robust version of $\Pp$.
	\end{restatable}
	\begin{proof}
		We show that Algorithm~\ref{alg:general-aso} is an $(\alpha,\beta)$-ASO for the LP formulation of $\Pp_R$. Since the algorithm checks whether the problem constraints are violated first, we only have to show that if there is a weight function $w^*$ and a solution $\sol$ such that
		$$\sum_{e\in\Omega}w^*_ex_e>\alpha\sum_{e\in \sol}w^*_e+\beta r,$$
		then if $\alg_i$ is the solution whose existence is guaranteed by \Cref{lem:55}, we show that Line~\ref{line:aso-constr-return} returns a constraint violation during the iteration on $\alg_i$.
		
		Suppose that such $w^*$ and $\sol$ exist. Now, let $\bar w$ be a weight function such that for all $e\in\Omega$, $\bar w_e = \ell_e$ if $e\in S$ and $u_e$ otherwise. Then:
		$$\sum_{e\in\Omega}\bar w_ex_e=\sum_{e\in \sol}\ell_ex_e + \sum_{e\notin \sol} u_ex_e >  \alpha\sum_{e\in \sol}\ell_e+\beta r.$$
		
		This implies that the regret of $x$ against $\sol$ is at least $(\alpha-1)\sum_{e\in \sol}\ell_e+\beta r$.
		
		Now, let $\alg_j$ be the solution whose existence is guaranteed by \Cref{lem:55} on $w$ such that $w_e=u_ex_e + \ell_e(1-x_e)$ and $w'$ such that $w'_e = \ell_e(1-x_e)$. The regret of the fractional solution $x$ against $\sol$ is:
		$$\sum_{e\notin \sol} u_ex_e - \sum_{e\in \sol}\ell_e(1-x_e)=w(\Omega\setminus \sol)-\sum_{e\in \Omega} \ell_e(1-x_e).$$
		
		Therefore, with the earlier lower bound on the regret of $x$ against $\sol$, we get:
		$$w(\Omega\setminus \sol)-\sum_{e\in \Omega} \ell_e(1-x_e)>(\alpha-1)\sum_{e\in \sol}\ell_e+\beta r.$$
		
		Now, we decompose $w(\Omega\setminus \sol)=w(\Omega\setminus (\sol\cup \alg_j))+w(\alg_j\setminus \sol)$ and isolate the second term:
		\begin{align*}
			w(\alg_j\setminus S)&>(\alpha-1)\sum_{e\in \sol}\ell_e+\beta r-w(\Omega\setminus (\sol\cup\alg_j))+\sum_{e\in \Omega} \ell_e(1-x_e) \\
			&= (\alpha-1)\sum_{e\in \sol}\ell_e+\beta r-\sum_{e\in\Omega\setminus (\sol\cup\alg_j)}u_ex_e + \sum_{e\in \sol\cup\alg_j} \ell_e(1-x_e).
		\end{align*}
		
		Now, the fact that from \Cref{lem:55}, $w(\alg_j\setminus \sol)\leq \beta w(\sol\setminus \alg_j)$ means that:
		$$w(\sol\setminus \alg_j)>\frac1\beta\left[(\alpha-1)\sum_{e\in \sol}\ell_e+\beta r-\sum_{e\in\Omega\setminus (\sol\cup\alg_j)}u_ex_e + \sum_{e\in \sol\cup\alg_j} \ell_e(1-x_e)\right].$$
		
		We are interested in the regret of the fractional solution $x$ against $\alg_j$, and we can now lower-bound it:
		\begin{align*}
			&w(\Omega\setminus \alg_j) - \sum_{e\in\Omega}\ell_e(1-x_e) \\
			&=w(\Omega\setminus(\sol\cup\alg_j)) + w(\sol\setminus \alg_j) - \sum_{e\in\Omega}\ell_e(1-x_e) \\
			&>\sum_{e\in\Omega\setminus(\sol\cup\alg_j)}(u_ex_e+\ell_e(1-x_e))\\
			&\qquad+\frac1\beta\left[(\alpha-1)\sum_{e\in \sol}\ell_e+\beta r-\sum_{e\in\Omega\setminus (\sol\cup\alg_j)}u_ex_e + \sum_{e\in \sol\cup\alg_j} \ell_e(1-x_e)\right] - \sum_{e\in\Omega}\ell_e(1-x_e) \\
			&=r + \frac{\beta-1}\beta\sum_{e\in\Omega\setminus(\sol\cup\alg_j)}u_ex_e + \frac{\alpha-1}{\beta}\sum_{e\in \sol}\ell_e +\frac1\beta\sum_{e\in \sol\cup\alg_j} \ell_e(1-x_e)-\sum_{e\in \sol\cup\alg_j} \ell_e(1-x_e)\\
			&> r + \frac{\beta-1}\beta\sum_{e\in\Omega\setminus(\sol\cup\alg_j)}u_ex_e + \frac{\alpha-1-\beta}{\beta}\sum_{e\in \sol}\ell_e +\frac1\beta\sum_{e\in \sol\cup\alg_j} \ell_e(1-x_e)-\sum_{e\in \alg_j\setminus \sol} \ell_e(1-x_e).
		\end{align*}
		
		Finally, since $\frac{\alpha-1-\beta}{\beta}=\Lambda(d\Gamma'+\kappa+1+\eps)$, the difference-approximation guarantee of \Cref{lem:55} ensures that:
		$$\sum_{e\in \alg_j\setminus \sol} \ell_e(1-x_e)\leq \frac{\alpha-1-\beta}\beta\sum_{e\in S}\ell_e(1-x_e)\leq \frac{\alpha-1-\beta}\beta\sum_{e\in \sol}\ell_e.$$
		
		Therefore, the regret of the fractional solution $x$ against $\alg_j$ on $w$ is lower-bounded as follows:
		\begin{align*}
			&w(\Omega\setminus \alg_j) - \sum_{e_in\Omega}\ell_e(1-x_e) \\
			&\qquad > r + \frac{\beta-1}\beta\sum_{e\in\Omega\setminus(\sol\cup\alg_j)}u_ex_e + \frac1\beta\sum_{e\in \sol\cup\alg_j} \ell_e(1-x_e) \\
			&\qquad > r.
		\end{align*}
		
		This concludes, and shows that the algorithm \textsc{General}-$\Pp_R$-\textsc{aso} returns a violated constraint.
	\end{proof}
	
	We show in the following \Cref{sec:apx-params} that we can always take $\Gamma=\Gamma'=17$ and $\kappa=d$ in this statement; these are the values used in \Cref{thm:main}.
	
	\subsection{Simple bounds for $\Gamma,\Gamma',\kappa$}\label{sec:apx-params}
	In order to make the results more practical and comparable, we want to provide some fixed values of $\Gamma,\Gamma'$ and $\kappa$ which satisfy the necessary conditions, which are as follows. We want:
	$$\min\left\{\frac{d\Gamma - 1}{2d\Gamma} ,\frac{(d\Gamma - 1)(\sqrt{\Gamma}-1)(\sqrt{\Gamma}-1-\eps)\kappa}{d^2(1+\eps)\Gamma^2}\right\} - (e^{\zeta' (d\Gamma'+\kappa+1+\eps)} - 1) > 0\,,$$
	where $\zeta' = \frac{d(1+\eps)\Gamma'}{(\sqrt{\Gamma'}-1)(\sqrt{\Gamma'}-1-\eps)(d\Gamma'-1)(d\Gamma-1)}.$ We also know from \Cref{lemma:algconverges} that $0 < \eps < \frac23 -\frac5{3d\Gamma}$. First, to simplify the expression, we will take $\eps=\frac15$, which is valid as soon as $\Gamma\geq 4$, and we get:
	$$\min\left\{\frac{d\Gamma - 1}{2d\Gamma},\frac{(d\Gamma - 1)(\sqrt{\Gamma}-1)(\sqrt{\Gamma}-\frac65)\kappa}{d^2\frac65\Gamma^2}\right\} - (e^{\zeta' (d\Gamma'+\kappa+\frac65)} - 1) > 0,$$
	with:
	$$\zeta' = \frac{d\frac65\Gamma'}{(\sqrt{\Gamma'}-1)(\sqrt{\Gamma'}-\frac65)(d\Gamma'-1)(d\Gamma-1)}.$$
	
	In order to minimize $\max(\Gamma,\Gamma')$, we can simply take $\Gamma=\Gamma'$, and furthermore, we set $\kappa=d$ (at least one of $\Gamma,\Gamma',\kappa$ needs to be a function of $d$). Computationally, we can see that the resulting condition that we need is an non-decreasing function of $d$. Therefore, we simply need to verify that it is positive for $d=1$, and we are done. The minimal value of $\Gamma=\Gamma'$ such that this is the case is $17$ (obtained computationally). Thus, we can write our approximation results directly using $\Gamma=\Gamma'=17$, and $\kappa=d$.
	
	\section{Application to Weighted $k$-Set Cover}\label{sec:appsc}
	
	In this section, we will show that the local search-based $H_k$-approximation algorithm for Weighted $k$-Set Cover introduced in \cite{GLLLocalSearchSetCover} can be modified slightly to satisfy the prerequisites, and therefore that the Weighted $k$-Set Cover problems admits a robust approximation algorithm. Indeed, the way that this algorithm was originally written uses the fact that approximately solving an instance of the Set Cover problem on a set of sets $\Omega$ can be done by approximately solving the problem on the downward-closure $\bar\Omega=\{A|B\in\Omega, A\subs B\}$, where a set $s\in\bar\Omega\setminus\Omega$ has weight the minimal weight of a superset of it in $\Omega$. Then, one can simply project a solution back to $\Omega$ by replacing any set of $\bar\Omega\setminus \Omega$ in the solution with its minimal-weight superset. However, this is not true in the robust setting which we are studying, since we need all the elements that we can select to have independent weights. This means that if we were to consider $\bar\Omega$ instead of $\Omega$ in the instance where the weight are in $[\ell,u]$, we could have a set $s\in\bar\Omega\setminus\Omega$ that is valued according to $\ell$ while it's minimal-cost superset could be valued according to $u$. In such a case, projecting a solution of $\bar\Omega$ that has a robust approximation guarantee to $\Omega$ like in the standard, non-robust case, can mean that some elements will be priced according to $u$ when we thought they were priced according to $\ell$. This can of course cause an arbitrary increase in the cost of the solution.
	
	In order to avoid this problem, we will rewrite the local search algorithm slightly by always working on $\Omega$ while implicitly maintaining a solution in $\bar\Omega$ which we will use to get the approximation guarantee. The algorithm works on a solution state $(S,D)$ where $S\subs \Omega$ is a valid set cover, and $D$ is a function mapping every set of $S$ to a set of $\bar\Omega$, such that $\{D_s\mid s\in S\}$ is a partition of $\bigcup\Omega$.
	
	An instance is a set $\Omega$ of sets, along with a weight function $w:\Omega\rightarrow\RR_+$. The local search algorithm works on a potential $\Phi$ which is defined for some solution state $(S,D)$ as:
	$$\Phi(S)=\sum_{s\in S} H_{|D(s)|}w_s,$$
	where we recall that $H_k$ is the $k$-th harmonic number $\sum_{i=1}^k\frac1k$. The local search starts from some solution state $(S_0,D_0)$. As soon as we have a solution $S_0$, it is easy to find some suitable $D_0$ by assigning every element $x\in\bigcup\Omega$ to be covered to some $s_x\in S_0$ such that $x\in s_x$. Then, it can be verified easily that $D_0: s\in S_0\mapsto \{x|s_x=s\}$ satisfies the required properties.
	
	Now, we are ready to describe the moves of the algorithm. When the algorithm is at a solution state $(S_i, D_i)$, each subset $Z\subs\bar\Omega$ of cardinality at most $k$ and containing only pairwise disjoint subsets corresponds to one move. For each such $Z$, there is a move which transforms $(S_i, D_i)$ to $(S'_i, D'_i)$ as follows. For each $s'\in Z$, let $t_{s'}\in\Omega$ be the minimal-weight superset of $s'$ in $\Omega$. We can then set:
	$$S'_i = \left\{t_{s'}\mid s'\in Z\right\}\cup \left\{s\mid s\in S_i, D_s\subsetneq \bigcup_{s'\in Z} s'\right\}.$$
	
	Then, for each $s'\in Z$, set $D'_i(t_{s'})=s'$ and otherwise, let $D'_i(s)=D_i(s)$ for any other set $s\in S'_i$. For shortened notations, we write that $S'_i=A_{Z}(S_i, D_i)$ and $D'_i=B_{Z}(S_i,D_i)$, and we say that $(S'_i, D'_i)$ is the move \emph{generated by} $Z$.
	
	Therefore, the set of moves available in a solution state $(S_i, D_i)$ is:
	$$\Pi_{(S_i,D_i)} = \left\{(A_{Z}(S_i, D_i), B_{Z}(S_i, D_i))\mid Z\subs\bar\Omega, |Z|\leq k,  \forall a,b\in Z, a\cap b=\es\right\}.$$

	Remark that this set has size at most ${|\bar\Omega|\choose k}\leq (2^kn)^k=O(n^k)$. Finally, on a solution state $(S_i, D_i)$, the algorithm chooses a move $(S_{i+1}, D_{i+1})\in\Pi_{(S_i,D_i)}$ such that $\Phi(S_{i+1},D_{i+1})\leq \Phi(S_i,D_i)$, if such a move exists. In order to guarantee polynomial-time termination, it suffices to stop as soon as the improvement in $\Phi$ is not better than some constant fraction of the potential.
	
	The proof that every local optimum is an $H_k$-approximation from \cite{GLLLocalSearchSetCover} still stands here, since we only add possible moves. Furthermore, we have the added benefit that the algorithm maintains at all times a solution on $\Omega$ which is controlled because it has the same weight and potential as $\{D_i(s)\mid s\in S_i\}\subs\bar\Omega$, which is the solution that \cite{GLLLocalSearchSetCover} shows has the approximation guarantee.
	
	\paragraph*{Well-behaved potential.}
	Given a weight function $w$ on $\Omega$, the local search algorithm uses the following potential function for some solution state $(S,D)$: $\Phi(S,D) = \sum_{s\in S} w(s)H_{|s|}$, therefore we set for all $s\in \Omega$, $\phi(s)=w(s)H_{|s|}$ and the potential can be expressed as $\Phi(S,D)=\sum_{s\in S} \phi(s)$, so it is linear. Furthermore, remark that for any solution $S$, we have: $w(S)\leq \Phi(S) \leq H_k w(S)$, and thus \Cref{ppt:bounded-potential} is satisfied for $\Lambda_1=1, \Lambda_2=H_k$.
	
	\paragraph*{Minimal Move Covering.}
	Let us fix two solution states $(S, D)$ and $(T,E)$. The set
	$$\Pi_{(S,D)}^{(T,E)}=\left\{(S',D')\mid (S',D')\in\Pi_{(S,D)}, \es\subsetneq S'\setminus S\subs T\right\}$$
	contains only moves $(S',D')$ such that $S'\setminus S$ contains between $1$ and $k$ sets of $T\setminus S$. Indeed, by definition of the moves, $|S'\setminus S|\leq k$. Therefore, these moves correspond to sets $S'\subs\bar\Omega$ of size $|S'|\leq k$ consisting of pairwise disjoint sets whose minimal-weight supersets are all in $T\setminus S$. 
	
	We define $\alpha$ over $\Pi_{(S,D)}^{(T,E)}$ as follows. For each set $s\in S\setminus T$, there exists a set $E_s\subs T$ which covers the elements of $D(s)$. Furthermore, since $|D(s)|\leq k$, we can find a set $E_s$ of cardinality at most $k$. The move $(S',D')$ generated by $E_s$ is indeed in $\Pi_{(S,D)}^{(T,E)}$ since $S'\setminus S=E_S\subs T$; we set $\alpha(S',D')=1$. Remark in particular that $s\notin S'$. Once we have set the $\alpha$-value of these $|S\setminus T|$ moves to $1$, we set all the others to $1$. We will show that this function $\alpha$ satisfies the minimal move covering property.
	
	First, let $s\in S\setminus T$. Then the sum
	$$\sum_{\substack{(S',D')\in \Pi_{(S,D)}^{(T,E)}\\ s\notin S'}}\alpha(S',D')$$
	contains at least the move generated by $E_s$ as explained above, and since its $\alpha$-value is $1$, this satisfies the first property.
	
	Now, let $s'\in T\setminus S$. We want to upper-bound the sum:
	$$\sum_{\substack{(S',D')\in \Pi_{(S,D)}^{(T,E)}\\ s'\in S'\setminus S}} \alpha(S',D').$$
	
	To do so, remark that this is equivalent to counting how many sets $s\in S\setminus T$ have $s'\in E_s$. Since $\{D(s)\mid s\in S\setminus T\}$ is a partition, a specific set $s'\in S\setminus T$ can only be used at most $k$ times to cover items of some $D(s)$, since it only contains $k$ different items. Thus:
	$$\sum_{\substack{(S',D')\in \Pi_{(S,D)}^{(T,E)}\\ s'\in S'\setminus S}} \alpha(S',D')\leq k$$
	
	Finally, remark that the total mass of $\alpha$ is at most $|S\setminus T|\leq n$ as remarked above. Therefore, this function $\alpha$ satisfies the $k$-MMC property.
	
	\paragraph*{Efficient Good Move Retrieval.}
	Finally, we show that the efficient good move retrieval property is satisfied (in the following, we omit the second half a solution state for legibility, and keep only the set cover in $\Omega$):

	\begin{lemma}
		Given two potentials $\Phi$ and $\Phi'$, a solution $S$ and some bound $W' > 0$, assume there is a move $S'\in\Pi_S$ such that $\Phi(S'\setminus S)-\Phi(S\setminus S')\leq W$ and has $\Phi'(S'\setminus S)-\Phi'(S\setminus S')\leq W'$. Then, there is a polynomial-time algorithm that finds a move $S''$ such that:
		\begin{itemize}
			\item $\Phi(S''\setminus S)-\Phi(S\setminus S'')\leq W$ and
			\item $\Phi'(S''\setminus S)-\Phi'(S\setminus S'')\leq W'$.
		\end{itemize}
	\end{lemma}
	\begin{proof}
		Since $|\Pi_S|\leq {n2^k\choose k} \leq 2^{k^2}n^k=O(n^k)$ (recall that $k$ is a constant), we can simply iterate in polynomial-time over every set $S''\in\Pi_S$ and we will find either $S'$ or a set with the same properties.
	\end{proof}
	
	\paragraph*{Robust approximation algorithm.}
	
	From the prerequisites, \Cref{thm:main} implies the following result as the standard LP relaxation has integrality gap at most $H_k$ and there exists an $H_k$-approximation algorithm \cite{GLLLocalSearchSetCover}:
	
	\thmSetcover
	
	\section{Conclusion}
	
	In this paper, we have developed some proofs found in the first paper about robust approximation algorithms for NP-hard problems \cite{GaneshRobustAlgorithmsTSP2023}. We believe that there are still many questions left to be answered. We currently do not know of NP-hard problems that admit an approximation algorithm but no robust approximation algorithm. Showing either that they exist, or that there is a way to transform any approximation algorithm to a robust approximation algorithm would be a big step toward understanding how robust approximation behaves.
	
	A major obstacle to this is that the robust approximation setting is less general that it seems at first glance. Indeed, it only works on problems that admit an approximation algorithm in the non-metric case (since the weights vary independently) and such that the weight is not involved in the feasibility constraints (this disqualifies some scheduling problems, for instance). Similarly, our more general robust approximation theorem (\Cref{thm:main}) is difficult to apply because these constraints on the kind of problem that can be studied in this robust approximation framework are limiting, and local search with approximation guarantees are somewhat rare when we cannot assume that we work on a metric.
	
	However, in light of the question of knowing whether any approximation algorithm can be turned to a robust approximation algorithm, it seems that we are simply lacking some important tools to release some of the assumptions that we currently need, and some more interesting work needs to be done before the question can be fully solved.

	\newpage
	\appendix
		
	\section{Proof of \Cref{lem:55} and necessary lemmas}
	
	\lemFiveFourteen*
	\begin{proof}
		
		We analyze one forward phase of the \textsc{DoubleApprox} algorithm. We let $\alg_{\rho,j}^{(i+1)}$ denote the value of the variable $\alg_\rho^{(i+1)}$ after $j$ iterations of the loop on Line~\ref{line:firstwhile}, and we let $J$ be the total number of iterations. Therefore, the final value of $\alg^{(i)}$ is equal to $\alg_{\rho,0}^{(i+1)}$. Because of the loop condition and the assumption that $\rho\leq(1+\eps)\Phi(\alg^{(i)}\setminus \sol)$, we have:
		
		$$\Phi(\alg^{(i+1)}_\rho) > \Phi(\alg^{(i)}) - \rho/2 \geq  \Phi(\alg^{(i)}) - \frac{1+\eps}{2}\Phi(\alg^{(i)} \setminus \sol).$$
		
		Then, by \Cref{lem512} and the assumption that $\Phi(\alg^{(i)} \setminus \sol) > d\Gamma \cdot \Phi(\sol \setminus \alg^{(i)})$ in the lemma statement, we have:
		
		\begin{align*}
			&\Phi(\alg^{(i+1)}_\rho \setminus \sol) \geq \Phi(\alg^{(i+1)}_\rho \setminus \sol) - \Phi(\sol \setminus \alg^{(i+1)}_{\rho}) \\
			\qquad&=\Phi(\alg^{(i)} \setminus \sol) - \Phi(\sol \setminus \alg^{(i)}) + \Phi(\alg^{(i+1)}_\rho) - \Phi(\alg^{(i)}) \text{ (\Cref{lem512})}\\
			&\geq \Phi(\alg^{(i)} \setminus \sol) - \Phi(\sol \setminus \alg^{(i)}) - \frac{1+\eps}{2}\Phi(\alg^{(i)} \setminus \sol) \\
			& \geq \left(\frac{1-\eps}{2} - \frac{1}{d\Gamma}\right)\Phi(\alg^{(i)} \setminus \sol) \\
			& \geq \frac1{10}\rho \text{ (since $\eps\leq \frac23-\frac5{3d\Gamma}$)}\,.
		\end{align*}
		
		Because of the assumption that $\Phi(\alg^{(i+1)}_\rho \setminus \sol) > d\Gamma \cdot \Phi(\sol \setminus \alg^{(i+1)}_{\rho})$ in the lemma statement, \Cref{lem510sc} guarantees that \textsc{GreedySwap} will find a good move. Moreover, since $\frac1{n^2}\Phi(\alg^{(i+1)}_\rho \setminus \sol) \geq \frac1{10n^2}\rho$, this move is accepted in Line~\ref{line:move-improves-enough} of \textsc{GreedySwap}, therefore $stop$ is never set to $1$. However, the loop will end in polynomial time because each move performed by \textsc{GreedySwap} improves the potential by at least $\frac1{10n^2}\rho>0$.
		
		\bigskip
		Now, suppose $\alg^{(i+1)}_{\rho, J}$ satisfies $\Phi(\alg^{(i+1)}_{\rho, J}) \leq \Phi(\alg^{(i)}) - \rho/2$, meaning that the while loop at Line~\ref{line:firstwhile} of \textsc{DoubleApprox} exits and ends the forward phase. In that case:
		
		\begin{align*}
			\Phi(\alg^{(i)}) - \Phi(\alg^{(i+1)}_{\rho, J}) &\geq \rho/2 \geq \frac{1}{2} \Phi(\alg^{(i)} \setminus \sol)\\
			&\geq \frac{1}{2} [\Phi(\alg^{(i)} \setminus \sol) - \Phi(\sol \setminus \alg^{(i)})]\\
			&= \frac{1}{2} [\Phi(\alg^{(i+1)}_{\rho,0} \setminus \sol) -\Phi(\sol \setminus \alg^{(i+1)}_{\rho,0})]\\
			&\geq \frac{1}{2}[\Phi(\alg^{(i+1)}_{\rho,J} \setminus \sol) - \Phi(\sol \setminus \alg^{(i+1)}_{\rho,J}) ]\\
			&\geq \frac{d\Gamma - 1}{2d\Gamma} \Phi(\alg^{(i+1)}_{\rho, J}\setminus \sol)\,.
		\end{align*}
		
		The second-to-last inequality uses \Cref{lem512}, which implies $\Phi(\alg^{(i+1)}_{\rho, j} \setminus \sol) - \Phi(\sol \setminus \alg^{(i+1)}_{\rho, j})$ is decreasing with swaps, and the last inequality holds by the assumption $\Phi(\alg^{(i+1)}_\rho \setminus \sol) > d\Gamma \cdot \Phi(\sol \setminus \alg^{(i+1)}_{\rho})$ in the lemma statement. Thus if $\Phi(\alg^{(i+1)}_{\rho, J}) \leq \Phi(\alg^{(i)}) - \rho/2$, the lemma holds.
		
		\bigskip
		Now, the loop might otherwise end with $\Phi(\alg^{(i+1)}_{\rho, J}) > \Phi(\alg^{(i)}) - \rho/2$, meaning that $\Phi'(\alg^{(i+1)}_{\rho, J}) > (d\Gamma' + \kappa)\chi'$. We want a lower bound on
		$$\Phi(\alg^{(i+1)}_{\rho, 0}) - \Phi(\alg^{(i+1)}_{\rho, J})= \sum_{j = 0}^{J-1}[\Phi(\alg^{(i+1)}_{\rho, j}) - \Phi(\alg^{(i+1)}_{\rho, j+1})]\,.$$
		
		We bound each $\Phi(\alg^{(i+1)}_{\rho, j}) - \Phi(\alg^{(i+1)}_{\rho, j+1})$ term using both \Cref{lem510sc} and Property~\ref{efficient-gm}. By \Cref{lem510sc} and the assumption in the lemma statement that  $\Phi(\alg^{(i+1)}_\rho \setminus \sol) > d\Gamma \cdot \Phi(\sol \setminus \alg^{(i+1)}_{\rho})$, we know there exists a move $S\in\Pi_\apij^{\sol}$ such that:
		
		$$\frac{(1+\eps)\Phi'(S\setminus \apij)-\Phi'(\apij\setminus S)}{\Phi(\apij\setminus S)-(1+\eps)\Phi(S\setminus \apij)} \leq \frac{d(1+\eps)\Gamma}{(\sqrt\Gamma-1)(\sqrt\Gamma-1-\eps)} \frac{\Phi'(\sol \setminus \apij)}{\Phi(\apij \setminus \sol)},$$
		
		where $\Phi(\apij\setminus S)-(1+\eps)\Phi(S\setminus \apij) > 0$.
		
		When \textsc{GreedySwap} is in the inner loop where $W$ is a $\sqrt{1+\eps}$-guess of $\Phi(S\setminus\apij)-\Phi(\apij\setminus S)$ and $W'$ is a $\sqrt{1+\eps}$-guess of $\Phi'(S\setminus\apij)-\Phi'(\apij\setminus S)$, Property~\ref{efficient-gm} applied with precision $\sqrt{1+\eps}-1>0$ (meaning that the $(1+\eps)$ in the statement of this property becomes $\sqrt{1+\eps}$ here) finds a move $S'\in\Pi_\apij^{\sol}$ such that:
		\begin{itemize}
			\item $\Phi'(S'\setminus\apij)-\Phi'(\apij\setminus S')\leq (1+\eps)\Phi'(S\setminus\apij)-\Phi'(\apij\setminus S)$, and
			\item $\Phi(S'\setminus\apij)-\Phi(\apij\setminus S')\leq (1+\eps)\Phi(S\setminus\apij)-\Phi(\apij\setminus S)$.
		\end{itemize}
		
		Therefore, the move $S^*\in\Pi_\apij^\sol$ chosen by the $(j+1)$-th call to \textsc{GreedySwap} satisfies:
		
		\begin{align*}
			\frac{\Phi'(S^*\setminus \apij)-\Phi'(\apij\setminus S^*)}{\Phi(\apij\setminus S^*)-\Phi(S^*\setminus\apij)}
			&\leq \frac{\Phi'(S'\setminus\apij)-\Phi'(\apij\setminus S')}{\Phi(\apij\setminus S')-\Phi(S'\setminus\apij)} \\
			&\leq \frac{(1+\eps)\Phi'(S\setminus\apij)-\Phi'(\apij\setminus S)}{\Phi(\apij\setminus S)-(1+\eps)\Phi(S\setminus\apij)} \\
			&\leq \frac{d(1+\eps)\Gamma}{(\sqrt\Gamma-1)(\sqrt\Gamma-1-\eps)} \frac{\Phi'(\sol \setminus \apij)}{\Phi(\apij \setminus \sol)}
		\end{align*}
		
		Rearranging terms and observing that $\Phi'(\sol) \geq \Phi'(\sol \setminus \alg^{(i+1)}_{\rho,j})$ gives:
		
		\begin{align*}
			&\Phi(\alg^{(i+1)}_{\rho, j}) - \Phi(\alg^{(i+1)}_{\rho, j+1})
					= \Phi(\apij\setminus S^*) - \Phi(S^*\setminus\apij)\\
			\qquad&\geq \frac{(\sqrt{\Gamma}-1)(\sqrt{\Gamma}-1-\eps)}{d(1+\eps)\Gamma} \cdot \Phi(\alg^{(i+1)}_{\rho,j} \setminus \sol) \frac{\Phi'(S^*\setminus\apij) - \Phi'(\apij\setminus S^*)}{\Phi'(\sol)}\\
			&= \frac{(\sqrt{\Gamma}-1)(\sqrt{\Gamma}-1-\eps)}{d(1+\eps)\Gamma} \cdot \Phi(\alg^{(i+1)}_{\rho,j} \setminus \sol) \frac{\Phi'(\alg^{(i+1)}_{\rho,j+1})-\Phi'(\alg^{(i+1)}_{\rho,j})}{\Phi'(\sol)}\,.
		\end{align*}
		
		We are now ready to sum this over all iterations of the forward phase:
		
		\begin{align*}
			\Phi&(\alg^{(i+1)}_{\rho, 0}) - \Phi(\alg^{(i+1)}_{\rho,J})\\
			&= \sum_{j = 0}^{J-1}[\Phi(\alg^{(i+1)}_{\rho, j}) - \Phi(\alg^{(i+1)}_{\rho, j+1})] \\
			&\geq \sum_{j=0}^{J-1} \frac{(\sqrt{\Gamma}-1)(\sqrt{\Gamma}-1-\eps)}{d(1+\eps)\Gamma} \cdot \Phi(\alg^{(i+1)}_{\rho,j} \setminus \sol) \frac{\Phi'(\alg^{(i+1)}_{\rho, j+1})-\Phi'(\alg^{(i+1)}_{\rho, j})}{\Phi'(\sol)} \\
			& \geq \frac{(\sqrt{\Gamma}-1)(\sqrt{\Gamma}-1-\eps)}{d(1+\eps)\Gamma} \sum_{j=0}^{J-1} [\Phi(\alg^{(i+1)}_{\rho, j} \setminus \sol) - \Phi(\sol \setminus \alg^{(i+1)}_{\rho, j})]\\
			&\qquad \cdot \frac{\Phi'(\alg^{(i+1)}_{\rho, j+1})-\Phi'(\alg^{(i+1)}_{\rho, j})}{\Phi'(\sol)} \\
			&\geq \frac{(\sqrt{\Gamma}-1)(\sqrt{\Gamma}-1-\eps)}{d(1+\eps)\Gamma} \sum_{j=0}^{J-1} [\Phi(\alg^{(i+1)}_{\rho, J} \setminus \sol) - \Phi(\sol \setminus \alg^{(i+1)}_{\rho, J})]\\
			&\qquad \cdot \frac{\Phi'(\alg^{(i+1)}_{\rho, j+1})-\Phi'(\alg^{(i+1)}_{\rho, j})}{\Phi'(\sol)}\\
			&\geq \frac{(d\Gamma - 1)(\sqrt{\Gamma}-1)(\sqrt{\Gamma}-1-\eps)}{d^2(1+\eps)\Gamma^2} \frac{\Phi(\alg^{(i+1)}_{\rho, J} \setminus \sol)}{\Phi'(\sol)}\\
			&\qquad \cdot \sum_{j=0}^{J-1}  \left[\Phi'(\alg^{(i+1)}_{\rho, j+1})-\Phi'(\alg^{(i+1)}_{\rho, j})\right]\\
			&=\frac{(d\Gamma - 1)(\sqrt{\Gamma}-1)(\sqrt{\Gamma}-1-\eps)}{d^2(1+\eps)\Gamma^2} \Phi(\alg^{(i+1)}_{\rho, J} \setminus \sol)  \frac{\Phi'(\alg^{(i+1)}_{\rho, J})-\Phi'(\alg^{(i+1)}_{\rho, 0})}{\Phi'(\sol)}\\
			&\geq \frac{(d\Gamma - 1)(\sqrt{\Gamma}-1)(\sqrt{\Gamma}-1-\eps)\kappa}{d^2(1+\eps)\Gamma^2} \Phi(\alg^{(i+1)}_{\rho, J} \setminus \sol)\,.
		\end{align*}
		
		The third-to-last inequality uses \Cref{lem512}, which implies that $\Phi(\alg^{(i+1)}_{\rho, j} \setminus \sol) - \Phi(\sol \setminus \alg^{(i+1)}_{\rho, j})$ is decreasing with swaps. The second-to-last inequality uses the assumption $\Phi(\alg^{(i+1)}_\rho \setminus \sol) > d\Gamma \cdot \Phi(\sol \setminus \alg^{(i+1)}_{\rho})$ in the statement of the lemma. The last inequality uses the fact that the while loop in Line~\ref{line:firstwhile} of \textsc{DoubleApprox} terminates because $\Phi'(\alg^{(i+1)}_{\rho, J}) > (d\Gamma' + \kappa)\chi'$, and lines~\ref{line:approx} and~\ref{line:secondwhile} of \textsc{DoubleApprox} give that  $\Phi'(\alg^{(i+1)}_{\rho, 0}) \leq d\Gamma' \chi'$, and also because $\chi'\geq \Phi(\sol)$.
	\end{proof}
	
	\smallskip
	\lemFiveFifteen*
	\begin{proof}
		Because $\Phi'(\alg^{(i+2)}_\rho) \geq d\Gamma'\chi'$ in a backward phase and $\chi' \geq \Phi'(\sol)$,  \Cref{lem510sc} guarantees that each time \textsc{GreedySwap} is called in Line~\ref{line:secondswap} of \textsc{DoubleApprox}, at least one move is possible. Since $\phi'$ only takes values that are multiples of $\frac{\eps}{n}\chi'$, and the last argument to \textsc{GreedySwap} is $\frac{\eps}{n}\chi'$ (which lower bounds the decrease in $\Phi'(\alg^{(i+2)}_\rho)$ due to any improving move), \textsc{GreedySwap} always makes a move.
		
		We let $\alg^{(i+2)}_{\rho, j}$ be the value of $\alg^{(i+2)}$ after $j$ calls to \textsc{GreedySwap} on $\alg^{(i+2)}$, and let $J$ be the total number of iterations of the loop on Line~\ref{line:secondwhile}. Then $\alg^{(i+2)}_{\rho, 0}$ is the final value of $\alg^{(i+1)}$ and the final value of $\alg^{(i+2)}$ is $\alg^{(i+2)}_{\rho, J}$.  We want to show the following upper bound:
		$$\Phi(\alg^{(i+2)}_{\rho, J}) - \Phi(\alg^{(i+2)}_{\rho, 0}) = \sum_{j = 0}^{J-1}[\Phi(\alg^{(i+2)}_{\rho, j+1}) - \Phi(\alg^{(i+2)}_{\rho, j})] \leq (e^{\zeta' T} - 1)  \Phi(\alg^{(i+2)}_{\rho, 0} \setminus \sol).$$
		
		For each $\Phi(\alg^{(i+2)}_{j+1}) - \Phi(\alg^{(i+2)}_{j})$ term, we use Lemma~\ref{lem510sc} and Lemma~\ref{efficient-gm}. Since  $\Phi'(\alg^{(i+2)}_\rho) > d\Gamma' \Phi'(\sol)$ in a backward phase, Lemma~\ref{lem510sc} implies that there exists some move $S\in\Pi_\aptj^\sol$ such that:
		$$\frac{(1+\eps)\Phi(S\setminus\aptj)-\Phi(\aptj\setminus S)}{\Phi'(\aptj\setminus S)-(1+\eps)\Phi'(S\setminus\aptj)} \leq \frac{d(1+\eps)\Gamma'}{(\sqrt{\Gamma'}-1)(\sqrt{\Gamma'}-1-\eps)} \frac{\Phi(\sol \setminus \aptj)}{\Phi'(\aptj \setminus \sol)},$$
		
		where $\Phi'(\aptj\setminus S)-(1+\eps)\Phi'(S\setminus\aptj)>0$.
		
		When \textsc{GreedySwap} is in the inner loop where $W$ is a $\sqrt{1+\eps}$-guess of $\Phi'(S\setminus\apij)-\Phi'(\apij\setminus S)$ and $W'$ is a $\sqrt{1+\eps}$-guess of $\Phi(S\setminus\apij)-\Phi(\apij\setminus S)$, \textsc{GreedySwap} uses \Cref{efficient-gm} applied with precision $\sqrt{1+\eps}-1>0$ (meaning that the $(1+\eps)$ in the statement of this property becomes $\sqrt{1+\eps}$ here) to find a move $S'\in\Pi_\aptj^\sol$ such that:
		\begin{itemize}
			\item $\Phi'(S'\setminus\aptj)-\Phi'(\aptj\setminus S')\leq (1+\eps)\Phi'(S\setminus\aptj)-\Phi'(\aptj\setminus S)$, and
			\item $\Phi(S'\setminus\aptj)-\Phi(\aptj\setminus S')\leq (1+\eps)\Phi(S\setminus\aptj)-\Phi(\aptj\setminus S)$.
		\end{itemize}
		
		Therefore, the move $S^*\in\Pi_\aptj^\sol$ returned by \textsc{GreedySwap} satisfies:
		
		\begin{align*}
			\frac{\Phi(S^*\setminus \aptj)-\Phi(\aptj\setminus S^*)}{\Phi'(\aptj\setminus S^*)-\Phi'(S^*\setminus\aptj)}
					&\leq \frac{\Phi(S'\setminus\aptj)-\Phi(\aptj\setminus S')}{\Phi'(\aptj\setminus S')-\Phi'(S'\setminus\aptj)} \\
			\qquad&\leq \frac{(1+\eps)\Phi(S\setminus\aptj)-\Phi(\aptj\setminus S)}{\Phi'(\aptj\setminus S)-(1+\eps)\Phi'(S\setminus\aptj)} \\
			&\leq \frac{d(1+\eps)\Gamma'}{(\sqrt\Gamma'-1)(\sqrt\Gamma'-1-\eps)} \frac{\Phi(\sol \setminus \aptj)}{\Phi'(\aptj \setminus \sol)} \\
			&\leq \frac{d(1+\eps)\Gamma'}{(\sqrt{\Gamma'}-1)(\sqrt{\Gamma'}-1-\eps)(d\Gamma'-1)d\Gamma} \frac{\Phi(\aptj \setminus \sol)}{\Phi'(\sol)}.
		\end{align*}
		
		The last inequality comes from the assumption that for all $j<J$, we have $\Phi(\alg^{(i+2)}_{\rho} \setminus \sol) > d\Gamma \cdot \Phi(\sol \setminus \alg^{(i+2)}_{\rho})$, and  $\Phi'(\alg^{(i+2)}_{\rho, j}) \geq d\Gamma' \Phi'(\sol) $ implies that
		$$\Phi'(\alg^{(i+2)}_{\rho, j} \setminus \sol) \geq \Phi'(\alg^{(i+2)}_{\rho, j}) - \Phi'(\sol) \geq (d\Gamma'-1) \Phi'(\sol).$$
		
		We can now compute:
		
		\begin{align*}
			&\Phi(\alg^{(i+2)}_{\rho, J}) - \Phi(\alg^{(i+2)}_{\rho, 0})\\ &= \sum_{j = 0}^{J-1}[\Phi(\alg^{(i+2)}_{\rho, j+1}) - \Phi(\alg^{(i+2)}_{\rho, j})] \nonumber \\
			&= \sum_{j = 0}^{J-1}\frac{\Phi(\alg^{(i+2)}_{\rho, j+1}) - \Phi(\alg^{(i+2)}_{\rho, j})}{\Phi'(\alg^{(i+2)}_{\rho, j}) - \Phi'(\alg^{(i+2)}_{\rho, j+1})} \cdot [\Phi'(\alg^{(i+2)}_{\rho, j}) - \Phi'(\alg^{(i+2)}_{\rho, j+1})] \nonumber \\
			&\leq  \frac{d(1+\eps)\Gamma'}{(\sqrt{\Gamma'}-1)(\sqrt{\Gamma'}-1-\eps)(d\Gamma'-1)d\Gamma} \sum_{j = 0}^{J-1} [\Phi(\alg^{(i+2)}_{\rho, j} \setminus \sol)] \nonumber\cdot \frac{\Phi'(\alg^{(i+2)}_{\rho, j}) - \Phi'(\alg^{(i+2)}_{\rho, j+1})}{\Phi'(\sol)} \nonumber \\
			&\leq \frac{d(1+\eps)\Gamma'}{(\sqrt{\Gamma'}-1)(\sqrt{\Gamma'}-1-\eps)(d\Gamma'-1)(d\Gamma-1)} \nonumber\cdot \sum_{j = 0}^{J-1}  [\Phi(\alg^{(i+2)}_{\rho, j} \setminus \sol) - \Phi(\sol \setminus \alg^{(i+2)}_{\rho, j})] \nonumber\\
			&\qquad\qquad  \cdot\frac{\Phi'(\alg^{(i+2)}_{\rho, j}) - \Phi'(\alg^{(i+2)}_{\rho, j+1})}{\Phi'(\sol)} \nonumber \\
			&= \zeta' \sum_{j = 0}^{J-1} [\Phi(\alg^{(i+2)}_{\rho, j} \setminus \sol) - \Phi(\sol \setminus \alg^{(i+2)}_{\rho, j})] \cdot \frac{\Phi'(\alg^{(i+2)}_{\rho, j}) - \Phi'(\alg^{(i+2)}_{\rho, j+1})}{\Phi'(\sol)}\,.
		\end{align*}
		
		The last inequality comes from the assumption that 
		\[
		\Phi(\alg^{(i+2)}_{\rho,j} \setminus \sol) > d\Gamma \cdot \Phi(\sol \setminus \alg^{(i+2)}_{\rho,j})
		\]
		for all $j\leq J$, which implies:
		\begin{align*}
			\Phi(\alg^{(i+2)}_{\rho, j} \setminus \sol) &= \frac{d\Gamma}{d\Gamma-1}\Phi(\alg^{(i+2)}_{\rho, j} \setminus \sol) - \frac{1}{d\Gamma-1}\Phi(\alg^{(i+2)}_{\rho, j} \setminus \sol)\\
			&< \frac{d\Gamma}{d\Gamma-1}\left[\Phi(\alg^{(i+2)}_{\rho, j} \setminus \sol) - \Phi(\sol \setminus \alg^{(i+2)}_{\rho, j})\right].
		\end{align*}

		It now suffices to show:
		
		$$\sum_{j = 0}^{J-1} [\Phi(\alg^{(i+2)}_{\rho, j} \setminus \sol) - \Phi(\sol \setminus \alg^{(i+2)}_{\rho, j})] \cdot \frac{\Phi'(\alg^{(i+2)}_{\rho, j}) - \Phi'(\alg^{(i+2)}_{\rho, j+1})}{\Phi'(\sol)} \leq $$
		$$\frac{e^{\zeta' T} - 1}{\zeta'} \Phi(\alg^{(i+2)}_{\rho, 0} \setminus \sol).$$ 
		
		This is done at the end of the proof of Lemma 5.15 in \cite{GaneshRobustAlgorithmsTSP2023} by analyzing the integral of a well-chosen function.
	\end{proof}
	
	\lemFiveSeventeen*
	\begin{proof}
		We let $\Delta(i) := \Phi(\alg^{(i)} \setminus \sol) - \Phi(\sol \setminus \alg^{(i)})$ for even $i$. We seek a contradiction by assuming that the assumptions of the lemma are true but that there is no intermediate value such as in the lemma. Then by the same argument as in the proof of \Cref{cor:swapgain}, for all $i$ we have:
		$$\Phi'(\alg^{(i)}_{\rho}) \leq (d\Gamma' +\kappa+ 1 + \eps) \Phi'(\sol).$$
		
		Therefore, if the lemma is false, we have that for all $i$ and $\rho$:
		$$\Phi(\alg^{(i)}_{\rho} \setminus \sol) > d\Gamma \Phi(\sol \setminus \alg^{(i)}_{\rho}).$$
		
		By \Cref{cor:swapgain}, and under the assumption of the lemma that
		$$\min\left\{\frac{d\Gamma - 1}{2d\Gamma} ,\frac{(d\Gamma - 1)(\sqrt{\Gamma}-1)(\sqrt{\Gamma}-1-\eps)\kappa}{d^2(1+\eps)\Gamma^2}\right\} - (e^{\zeta' (d\Gamma'+\kappa+1+\eps)} - 1) > 0,$$
		we know that for all $i$, $\Phi(\alg^{(i)})-\Phi(\alg^{(i+2)})\geq 0$ and therefore, the loop on Line~\ref{line:mainwhile} never breaks. We will reach a contradiction by showing that $\Delta(I)<\frac{d\Gamma-1}{d\Gamma}\min_{\substack{i\in\Omega\\\phi(i)>0}}\phi(i)$. Indeed, if this happens, we know from the proof of \Cref{lemma:backwardphase} how the inequality $\Phi(\alg^{(I)} \setminus \sol) > d\Gamma \Phi(\sol \setminus \alg^{(I)})$ implies that $\Phi(\alg^{(I)} \setminus \sol) - \Phi(\sol \setminus \alg^{(I)})>\frac{d\Gamma-1}{d\Gamma}\Phi(\alg^{(I)}\setminus\sol)$. However, since we do not have a $d\Gamma$-difference approximation, $\Phi(\alg^{(I)}\setminus\sol)>0$ which implies $\Phi(\alg^{(I)}\setminus\sol)>\min_{\substack{i\in\Omega\\\phi(i)>0}}\phi(i)$.
		
		Now, let us show how we reach this contradiction by studying how $\Delta(i)$ evolves. At first, $\Phi(0)\leq n\max_{i\in\Omega}\phi(i)$. Then, by \Cref{lem512}, we know that: $\Delta(i)-\Delta(i+2)=\Phi(\alg^{(i)})-\Phi(\alg^{(i+2)})$. Now, setting $\rho$ as in \Cref{lemma:forwardphase} such that $\rho \in [\Phi(\alg^{(i)} \setminus \sol), (1+\eps) \Phi(\alg^{(i)} \setminus \sol)]$, we have by \Cref{cor:swapgain}:
		
		\begin{align*}
			&\Delta(i)-\Delta(i+2) \\
			&\quad= \Phi(\alg^{(i)})-\Phi(\alg^{(i+2)}) \\
			&\quad\geq \left[\min\left\{\frac{d\Gamma - 1}{2d\Gamma} ,\frac{(d\Gamma - 1)(\sqrt{\Gamma}-1)(\sqrt{\Gamma}-1-\eps)\kappa}{d^2(1+\eps)\Gamma^2}\right\} - (e^{\zeta' (d\Gamma'+\kappa+1+\eps)} - 1)\right]\\
			&\quad \qquad  \cdot c(\alg^{(i+1)}_\rho \setminus \sol) \\
			&\quad\geq \left[\min\left\{\frac{d\Gamma - 1}{2d\Gamma} ,\frac{(d\Gamma - 1)(\sqrt{\Gamma}-1)(\sqrt{\Gamma}-1-\eps)\kappa}{d^2(1+\eps)\Gamma^2}\right\} - (e^{\zeta' (d\Gamma'+\kappa+1+\eps)} - 1)\right]\\
			&\quad \qquad\cdot [\Phi(\alg^{(i+1)}_\rho \setminus \sol) - \Phi(\sol \setminus \alg^{(i+1)}_\rho)].
		\end{align*}
		
		Since we have assumed that $\Phi(\alg^{(i+1)}_{\rho} \setminus \sol) > d\Gamma \Phi(\sol \setminus \alg^{(i+1)}_{\rho})$, we have the following:
		$$\Phi(\alg^{(i+1)}_{\rho} \setminus \sol) - \Phi(\sol \setminus \alg^{(i+1)}_{\rho})  > \frac{d\Gamma - 1}{d\Gamma} \Phi(\alg^{(i+1)}_{\rho} \setminus \sol).$$
		
		We use this inequality, along with both \Cref{lemma:backwardphase}, knowing that $T\leq d\Gamma'+\kappa+1+\eps$ (as shown in the proof of \Cref{cor:swapgain}), and \Cref{lem512}, to get the following derivation:
		
		\begin{align*}
			&\Delta(i+2) - [\Phi(\alg^{(i+1)}_\rho \setminus \sol) - \Phi(\sol \setminus \alg^{(i+1)}_\rho)] \\
			&\quad\leq (e^{\zeta' (d\Gamma'+\kappa+1+\eps)} - 1) \Phi(\alg^{(i+1)}_\rho \setminus \sol) \\
			&\quad< \frac{d\Gamma}{d\Gamma - 1} (e^{\zeta' (d\Gamma'+\kappa+1+\eps)} - 1) [\Phi(\alg^{(i+1)}_\rho \setminus \sol) - \Phi(\sol \setminus \alg^{(i+1)}_\rho)],
		\end{align*}
		
		which implies:
		$$\Phi(\alg^{(i+1)}_\rho \setminus \sol) - \Phi(\sol \setminus \alg^{(i+1)}_\rho) > \frac{1}{1 + \frac{d\Gamma}{d\Gamma - 1} (e^{\zeta' (d\Gamma'+\kappa+1+\eps)} - 1)}\Delta(i+2).$$
		
		Now, we merge this inequality into our lower bound of $\Delta(i)-\Delta(i+2)$:
		
		$$\Phi(i+2) < \left(1 +  \frac{\min\{\frac{d\Gamma - 1}{2d\Gamma} ,\frac{(d\Gamma - 1)(\sqrt{\Gamma}-1)(\sqrt{\Gamma}-1-\eps)\kappa}{d^2(1+\eps)\Gamma^2}\} - (e^{\zeta' (d\Gamma'+\kappa+1+\eps)} - 1)}{1 + \frac{d\Gamma}{d\Gamma - 1} (e^{\zeta' (d\Gamma +\kappa+1+\eps)} - 1)}\right)^{-1} \Phi(i) $$
		$$= (1+\eta)^{-1}\Phi(i).$$
		
		Therefore, $\Delta(i)<(1+\eta)^{-i/2}\Delta(0) \leq (1+\eta)^{-i/2}n\max_{i\in\Omega}\phi(i)$.
		
		Finally, setting $i=I$ gives $\Delta(I) < \min_{\substack{i\in\Omega\\\phi(i)>0}}\phi(i)$, yielding the contradiction explained above.
	\end{proof}
	
	\lemFiveFive*
	\begin{proof}
		One important remark here is that either $\Phi(\sol)$ or $\Phi'(\sol)$ might be $0$. If this is the case, the result is easy to get, and we treat it as a special case. Indeed, remark that since by definition $\Phi\geq \Phi'$, either (a) $\Phi'(\sol)=0$ and $\Phi(\sol)>0$ or (b) $\Phi'(\sol)=\Phi(\sol)=0$. Remark furthermore that we use a well-behaved potential, we know that for all $i$, $\phi(i)=0\iff w_i=0$ and $\phi'(i)=0\iff w'_i=0$.
		
		In case (a), $\Phi'(\sol)=0$ means that $w'(\sol)=0$. Therefore, to get the approximation guarantee that we want on $w'$, we need to find a solution $\alg_j$ with $w'(\alg_j)=0$. This means that we can simply restrict our attention to solutions in $\Omega_0\subs\Omega$, the set of items with $w'$-weight $0$. Then, we can simply run the local search algorithm on $w$, and \Cref{lem-ls-dapx} ensures that we get some solution that is a $\gamma$-difference-approximation on $w$ (where $\gamma$ is the approximation guarantee of the local search algorithm), thereby satisfying both approximation guarantees since $d\geq\lambda$.
		
		In case (b), we have $w(\sol)=w'(\sol)=0$. Therefore, the approximations guarantees that we want are as follows: $w'(\alg_j)=0$ and $w(\alg_j\setminus \sol)=0$. Since $w(\sol)=0$, the second guarantee implies that $w(\alg_j)=0$ too. So, we can once again restrict our attention to solutions in $\Omega_{00}\subs\Omega$, defined as the set of items $i\in\Omega$ such that $w(i)=w'(i)=0$. Any feasible solution in this set will satisfy the guarantees that we are interested in.
		
		This concludes the analysis in the case where $\Phi(\sol)=0$ or $\Phi'(\sol)=0$; therefore, we will now assume that both are positive. In particular, this means that $\Phi(\sol)\geq \min_{\substack{i\in\Omega\\\phi(i)\neq 0}}\phi(i)$, and $\Phi'(\sol)\geq \min_{\substack{i\in\Omega\\\phi'(i)\neq 0}}\phi'(i)$.
		
		We now assume that we have $\chi\in[\Phi(\sol), (1+\eps)\Phi(\sol)]$ and $\chi'\in[\Phi'(\sol), (1+\eps)\Phi'(\sol)]$ and that the potential $\Phi'$ has been properly quantified to only take values that are multiples of $\frac\eps n\chi'$, then the conditions of \textsc{DoubleApprox} are met. Now, \Cref{lemma:algconverges} applies as soon as:
		$$\min\left\{\frac{d\Gamma - 1}{2d\Gamma} ,\frac{(d\Gamma - 1)(\sqrt{\Gamma}-1)(\sqrt{\Gamma}-1-\eps)\kappa}{d^2(1+\eps)\Gamma^2}\right\} - (e^{\zeta' (d\Gamma'+\kappa+1+\eps)} - 1) > 0.$$
		
		We show that there exists some large enough $\Gamma>1$ such that this is the case. Indeed, since $\zeta'\xrightarrow[\Gamma\rightarrow\infty]{} 0$,
		$$\min\left\{\frac{d\Gamma - 1}{2d\Gamma} ,\frac{(d\Gamma - 1)(\sqrt{\Gamma}-1)(\sqrt{\Gamma}-1-\eps)\kappa}{d^2(1+\eps)\Gamma^2}\right\} - (e^{\zeta' (d\Gamma'+\kappa+1+\eps)} - 1)$$ $$\xrightarrow[\Gamma\rightarrow\infty]{} \min\left\{\frac12, \frac{\kappa}{d(1+\eps)}\right\}>0.$$
		
		Since $\Gamma>1$, asking for $\eps<1/4$ ensures that $\eps<2/3-\frac5{3d\Gamma}$. Therefore \Cref{lemma:algconverges} applies and there is some solution $\alg_j$ obtained by the algorithm \textsc{DoubleApprox} such that $\Phi(\alg_j\setminus \sol)\leq d\Gamma\cdot \Phi(\sol\setminus \alg_j)$ and $\Phi'(\alg_j)\leq (d\Gamma'+\kappa+1+\eps)\Phi'(\sol)$. Now, we use \Cref{ppt:bounded-potential} to conclude: $w(\alg_j\setminus \sol)\leq d\Gamma\Lambda\cdot w(\sol\setminus \alg_j)$ and $w'(\alg_j)\leq (d\Gamma'+\kappa+1+\eps)\Lambda w'(\sol)$.
		
		Therefore, the set $\alg$ of solutions that contains one solution with the guarantees of the lemma can be composed of the solution computed in cases (a) and (b), along with all the solutions returned by the \textsc{DoubleApprox} algorithm.
		
		Now, we show that we can indeed find this special solution $\alg_j$ in polynomial time. Let $m=\min_{\substack{i\in\Omega\\\phi(i)>0}}\phi(i)$ and $m'=\min_{\substack{i\in\Omega\\\phi'(i)>0}}\phi'(i)$. We can call \textsc{DoubleApprox} many times, for each pair $(\chi, \chi')$ where $$\chi\in\left\{m, \sqrt{1+\eps}m, ..., \sqrt{1+\eps}^{\left\lceil\log_{\sqrt{1+\eps}}\left(\frac{n\max_{i\in\Omega}\phi(i)}{m}\right)\right\rceil}m\right\}$$ and $$\chi'\in\left\{m', \sqrt{1+\eps}m', ..., \sqrt{1+\eps}^{\left\lceil\log_{\sqrt{1+\eps}}\left(\frac{n\max_{i\in\Omega}\phi'(i)}{m'}\right)\right\rceil}m'\right\}\,.$$ These two sets are polynomially sized. When we call \textsc{DoubleApprox} for some value $\chi'$, we round the potential $\Phi'$ up to smallest greater multiple of $\frac\eps n\chi'$ (that is, $\phi'(i)$ becomes $\left\lceil\frac{\phi'(i)}{\frac\eps n\chi'}\right\rceil\frac\eps n\chi'$); since the problem is linear, we just lose an $O(\eps)$ term in the approximation ratio.
		
		Finally, each call of \textsc{DoubleApprox} runs in polynomial time. We have shown in \Cref{lemma:algconverges} that the loop in Line~\ref{line:mainwhile} of \textsc{DoubleApprox} stops after a polynomial number of iterations. The forward loop in Line~\ref{line:firstwhile} cannot run for more than $5n^2$ iterations because each move decreases the potential $\Phi$ by at least $\frac1{10n^2}\rho$ and the loop ends if the decrease is more than $\rho/2$. The backward loop in Line~\ref{line:secondwhile} breaks when the potential $\Phi'$ reaches $d\Gamma'\chi'$, and as shown in the proof of \Cref{cor:swapgain}, the potential $\Phi'$ is at most $(d\Gamma'+\kappa+1+\eps)\chi'$ at the beginning of the loop. Therefore, since each iteration decreases the potential by at least $\frac\eps n\chi'$, there can be at most $(\kappa+1+\eps)\frac n\eps$ iterations.
	\end{proof}
	
\end{document}